\documentclass[journal]{IEEEtran}
\IEEEoverridecommandlockouts

\usepackage{amsmath,amssymb,amsfonts}
\usepackage{array}
\usepackage[caption=false,font=normalsize,labelfont=sf,textfont=sf]{subfig}
\usepackage{textcomp}
\usepackage{stfloats}
\usepackage{verbatim}
\usepackage{graphicx}
\usepackage{cite}
\usepackage{tabularx}
\usepackage{booktabs}
\usepackage{tikz}
\usetikzlibrary{arrows.meta,positioning,calc,backgrounds,shapes.geometric,decorations.markings}
\usepackage{pgfplots}
\pgfplotsset{compat=1.18}
\usepackage{algorithm}
\usepackage{algpseudocode}

\def\BibTeX{{\rm B\kern-.05em{\sc i\kern-.025em b}\kern-.08em
    T\kern-.1667em\lower.7ex\hbox{E}\kern-.125emX}}

\usepackage{acronym}
\acrodef{5G}{the fifth generation}
\acrodef{6G}{the sixth generation}
\acrodef{MIMO}{multiple-input multiple-output}
\acrodef{MISO}{multiple-input single-output}
\acrodef{MU}{multi-user}
\acrodef{EE}{energy efficiency}
\acrodef{SE}{spectral efficiency}
\acrodef{RF}{radio frequency}
\acrodef{BS}{base station}
\acrodef{UE}{user equipment}
\acrodef{DAC}{digital-to-analog converter}
\acrodef{PA}{power amplifier}
\acrodef{DSP}{digital signal processing}
\acrodef{SINR}{signal-to-interference-plus-noise ratio}
\acrodef{WMMSE}{weighted minimum mean-square error}
\acrodef{AWGN}{additive white Gaussian noise}
\acrodef{MiLAC}{microwave linear analog computer}
\acrodef{MU-MISO}{multiuser multiple-input single-output}
\acrodef{RIS}{reconfigurable intelligent surface}
\acrodef{BD-RIS}{beyond-diagonal reconfigurable intelligent surface}
\acrodef{BD}{block diagonalization}
\acrodef{QR}{QR}\acused{QR}
\acrodef{SNR}{signal-to-noise ratio}
\acrodef{CDF}{cumulative distribution function}
\acrodef{BCD}{block coordinate descent}
\acrodef{LMI}{linear matrix inequality}
\acrodef{CSI}{channel state information}
\acrodef{ZF}{zero forcing}
\acrodef{PIN}{p-i-n}
\acrodef{5GNR}[5G~NR]{5G New Radio}
\acrodef{SVD}{singular value decomposition}
\acrodef{i.i.d.}{independent and identically distributed}
\acrodef{UPA}{uniform planar array}
\acrodef{MRT}{maximum ratio transmission}
\acrodef{PSD}{positive semidefinite}
\acrodef{MSE}{mean-square error}
\acrodef{MMSE}{minimum mean-square error}
\acrodef{CG}{conjugate gradient}
\acrodef{MO-AltMin}{manifold-optimization alternating minimization}
\acrodef{EMF}{electromotive force}
\acrodef{IFT}{implicit function theorem}
\acrodef{KKT}[KKT]{Karush--Kuhn--Tucker}
\acrodef{PS}{phase shifter}
\acrodef{NLOS}{non-line-of-sight}
\acrodef{FC}{fully-connected}
\acrodef{SC}{stem-connected}
\acrodef{PHP}{post-hoc projection}
\acrodef{AR}{alternating refinement}
\acrodef{BSUM}{block successive upper-bound minimization}
\acrodef{SIM}{stacked intelligent metasurface}
\acrodef{CI}{confidence interval}
\usepackage{amsthm}
\newtheoremstyle{zackplain}%
  {\topsep}{\topsep}{\itshape}{}{\bfseries}{:}{.5em}{}
\theoremstyle{zackplain}

\newtheorem{proposition}{\textbf{Proposition}}

\newtheorem{definition}{\textbf{Definition}}
\newtheorem{remark}{\textbf{Remark}}

\renewenvironment{proof}{\par\noindent\textit{\textbf{Proof:}}\ }{\hfill$\square$\par\medskip}
\usepackage{color}
\usepackage[bottom=1.02in,top=0.75in,left=0.625in,right=0.625in]{geometry}
\usepackage{units}
\usepackage{bm}
\usepackage{amssymb}
\newcommand{\TT}{\mathsf{T}}
\newcommand{\HH}{\mathsf{H}}

\newcommand{\realp}[1]{\operatorname{Re}\!\left\{#1\right\}}

\newcommand{\bv}{{\bf b}}

\newcommand{\ev}{{\bf e}}
\newcommand{\fv}{{\bf f}}

\newcommand{\hv}{{\bf h}}

\newcommand{\pv}{{\bf p}}

\newcommand{\sv}{{\bf s}}

\newcommand{\uv}{{\bf u}}
\newcommand{\wv}{{\bf w}}
\newcommand{\vv}{{\bf v}}
\newcommand{\xv}{{\bf x}}

\newcommand{\zerov}{{\bf 0}}
\newcommand{\onev}{{\bf 1}}
\newcommand{\Am}{{\bf A}}
\newcommand{\Bm}{{\bf B}}

\newcommand{\Em}{{\bf E}}
\newcommand{\Fm}{{\bf F}}
\newcommand{\Gm}{{\bf G}}
\newcommand{\Hm}{{\bf H}}
\newcommand{\Id}{{\bf I}}

\newcommand{\Mm}{{\bf M}}
\newcommand{\Nm}{{\bf N}}

\newcommand{\Pm}{{\bf P}}

\newcommand{\Um}{{\bf U}}
\newcommand{\Wm}{{\bf W}}
\newcommand{\Vm}{{\bf V}}
\newcommand{\Xm}{{\bf X}}
\newcommand{\Ym}{{\bf Y}}

\newcommand{\Zerom}{{\bf 0}}

\newcommand{\omegav}{\hbox{\boldmath$\omega$}}

\newcommand{\Sigmam}{\hbox{\boldmath$\Sigma$}}

\newcommand{\Thetam}{\hbox{\boldmath$\Theta$}}
\newcommand{\Omegam}{\hbox{\boldmath$\Omega$}}
\newcommand{\Xim}{\hbox{\boldmath$\Xi$}}

\usepackage{amsmath}
\DeclareMathOperator*{\argmax}{arg\,max}
\DeclareMathOperator*{\argmin}{arg\,min}
\DeclareMathOperator{\tr}{tr}
\DeclareMathOperator{\diag}{diag}
\DeclareMathOperator{\rank}{rank}
\DeclareMathOperator{\St}{St}
\makeatletter
\let\zk@label\label
\renewcommand*{\label}[1]{\zk@label{#1}\ifvmode\if@nobreak\nobreak\fi\fi}
\makeatother

\begin{document}

\title{Beamforming Design for Stem-Connected Microwave Linear Analog Computer (MiLAC)-Aided Multiuser MISO Downlinks}

\author{Yuchen Zhang, \emph{Member, IEEE}, Zheyu Wu, \emph{Member, IEEE}, \\Bruno Clerckx, \emph{Fellow, IEEE}, and Tareq Y. Al-Naffouri, \emph{Fellow, IEEE}
\thanks{
Y. Zhang and T. Y. Al-Naffouri are with the Computer, Electrical, and Mathematical Science \& Engineering (CEMSE) Division, King Abdullah University of Science and Technology (KAUST), Thuwal 23955-6900, Kingdom of Saudi Arabia (e-mail: \{yuchen.zhang, tareq.alnaffouri\}@kaust.edu.sa).}
\thanks{Z. Wu and B. Clerckx are with the Department of Electrical and Electronic Engineering, Imperial College London, London, SW7 2AZ, U.K. (email: \{zheyu.wu, b.clerckx\}@imperial.ac.uk).}
}

\maketitle

\bstctlcite{IEEEexample:BSTcontrol}

\begin{abstract}
A microwave linear analog computer (MiLAC) is a tunable microwave network that leverages wave propagation to perform computation directly in the analog domain. MiLAC finds promising applications in beamforming, where it can map the data streams to the antennas without a digital beamformer, using one radio-frequency chain per stream. A fully-connected MiLAC needs a number of tunable components that grows quadratically with the number of antennas, whereas a stem-connected MiLAC needs only a linearly growing number and still achieves the point-to-point capacity. For the multiuser multiple-input single-output downlink, we prove that the stem-connected MiLAC attains the maximum sum-rate of the fully-connected MiLAC whenever $N\ge 2K-1$, where $N$ and $K$ denote the numbers of antennas and users. The proof shows that restricting the design to beamformers with orthonormal columns, which the stem-connected MiLAC realizes in closed form, loses no sum-rate. Based on this restriction, we develop, for every $N\ge K$, a weighted minimum mean-square error algorithm with a Riemannian beamformer update and prove that its accumulation points are stationary. For the bounded, discrete susceptances of practical components, we then propose an alternating refinement that adjusts one susceptance at a time through rank-one updates and never decreases the sum-rate. With $64$ antennas, four users, and ideal components, the stem-connected MiLAC attains the fully-connected sum-rate with $4.5$ times fewer tunable components. With $3$-bit susceptances at a signal-to-noise ratio of $10$~dB, the refinement retains $90\%$ of this sum-rate.
\end{abstract}


\section{Introduction}
\label{sec:intro}

Active \ac{RF} chains and high-resolution \acp{DAC} constitute a dominant per-antenna cost in the large arrays envisioned for fifth- and sixth-generation wireless networks~\cite{Bjornson2025Gigantic,Larsson2014Massive,Ning2026PMI}. A standard remedy is hybrid analog--digital beamforming. For $K$ data streams and $N$ antennas, it cascades a digital beamformer of size $N_{\mathrm{RF}}\times K$ with an analog beamformer of size $N\times N_{\mathrm{RF}}$, which reduces the number of active \ac{RF} chains from $N$ to $N_{\mathrm{RF}}$, where $K\le N_{\mathrm{RF}}<N$. A representative analog implementation is the \ac{PS} network, where each \ac{RF} chain is connected to every antenna through a tunable unit-modulus \ac{PS}~\cite{ElAyach2014Spatially,Sohrabi2016Hybrid,Yu2016AltMin}. This architecture is effective but subject to two restrictions. First, since each \ac{PS} has unit modulus, every entry of the analog beamformer has constant modulus, which limits the set of realizable beamformers: a PS-based hybrid beamformer can realize any fully digital beamformer with $N_{\mathrm{RF}}\ge 2K$~\cite{Sohrabi2016Hybrid}, whereas a gap remains for $K\le N_{\mathrm{RF}}<2K$~\cite{Yu2016AltMin}. Second, the baseband digital beamformer combines the data streams at the symbol rate, producing weighted sums of symbols with a larger dynamic range than the symbols themselves, which hinders the use of low-resolution \acp{DAC}.

A recent line of work on analog computing removes both restrictions of \ac{PS}-based hybrid beamforming at once. A \ac{MiLAC} is a tunable multiport microwave network whose ports are interconnected through reconfigurable admittances~\cite{Nerini2025AnalogI,Nerini2025AnalogII}. In a \ac{MiLAC}-aided transmitter with one \ac{RF} chain per user, $K$ ports are connected to the \ac{RF} chains and the remaining $N$ ports to the antennas, so the \ac{MiLAC} performs analog beamforming. The \ac{MiLAC} concept does not require a specific type of admittance, but most designs assume \emph{lossless} (purely imaginary) and \emph{reciprocal} (symmetric) admittances, which make the scattering matrix unitary and symmetric~\cite{Wu2026MiLAC,Nerini2025Capacity}. Losslessness is an idealization, and reciprocity avoids the non-reciprocal devices that asymmetric admittances would require. Under these assumptions, the analog beamformer from the \ac{RF} chains to the antennas is constrained only by a bound on its spectral norm~\cite{Wu2026MiLAC}, and its entries need not have constant modulus.

Two design choices distinguish existing \ac{MiLAC} beamforming schemes. The first is whether a digital beamformer precedes the \ac{MiLAC}. A hybrid digital-\ac{MiLAC} architecture, which feeds the \ac{MiLAC} through a $K\times K$ digital beamformer, can realize any digital beamformer~\cite{Wu2026MiLAC}, with half the \ac{RF} chains that a \ac{PS}-based hybrid beamformer needs. The pure-analog architecture drops the digital beamformer and sends the constellation symbols directly to the \ac{RF} chains. It avoids symbol-rate baseband processing and allows low-resolution \acp{DAC}, but cannot attain the fully digital sum-rate in general. The sum-rate gap narrows as the user channels become nearly orthogonal in large arrays~\cite{Wu2026MiLAC,Fang2026Performance}. More recently, a two-layer \ac{MiLAC} was shown to attain the fully digital sum-rate purely in the analog domain, without any digital beamformer~\cite{Zhou2026TwoLayer}.

The second choice is the interconnection topology. A \ac{FC} \ac{MiLAC} connects every pair of ports through a tunable component and requires $(N{+}K)(N{+}K{+}1)/2 = \mathcal{O}(N^2)$ tunable components, which is prohibitive at large $N$. To reduce this circuit complexity, the \ac{SC} \ac{MiLAC} was recently proposed. It requires only $K(2N{+}1) = \mathcal{O}(NK)$ tunable components and remains capacity-achieving in point-to-point \ac{MIMO} systems~\cite{Nerini2026Stem}. Other works address channel estimation~\cite{ZhangEstimation},~\cite{ZhangMUEstimation}, lossy tunable components~\cite{ZhouLossy}, mutual coupling~\cite{NeriniPhysics}, energy efficiency~\cite{Zack2026EE}, wideband systems~\cite{Peng2026OFDM}, and hardware realization with hybrid couplers and \acp{PS}~\cite{NeriniHardware}. Whereas a \ac{MiLAC} processes the signals before they reach the antennas, other emerging technologies exploit antenna- and wave-domain reconfigurability, each under its own hardware constraints. Examples are dynamic metasurface antennas~\cite{Shlezinger2021Dynamic,Ruizhi2026ICASSP}, reconfigurable and holographic antennas~\cite{Zhuoyang2026TSP,Wenyan2026ST}, and movable and fluid antennas, which reconfigure the antenna positions~\cite{Ruizhi2026TWC,Yichi2025TVT} or radiation patterns~\cite{Alireza2026JSTSP}, including fluid antenna systems for cell-free networks~\cite{Du2026FAS}. In \acp{SIM}~\cite{An2023SIM,Zhang2026Spherical,Wei2026NearField}, several cascaded layers of programmable metasurfaces process the signals in the wave domain.

Two directions identified as open in~\cite{Nerini2026Stem} motivate this work: whether the \ac{SC}-\ac{MiLAC} remains optimal in multiuser systems and how to optimize it there, and how \acp{MiLAC} perform when the tunable components are restricted to a limited range. With a digital beamformer, the first question has a direct answer: the factorization of~\cite{Wu2026MiLAC} splits any digital beamformer into an analog factor with orthonormal columns, which an \ac{SC}-\ac{MiLAC} realizes~\cite{Nerini2026Stem}, and a $K\times K$ digital factor (Remark~\ref{rem:hybrid}). Without a digital beamformer, the answer does not follow from~\cite{Wu2026MiLAC} and~\cite{Nerini2026Stem}. The \ac{FC}-\ac{MiLAC} algorithm of~\cite{Wu2026MiLAC} returns beamformers in the range of the channel matrix, which need not have orthonormal columns. The closed-form synthesis of~\cite{Nerini2026Stem} does not apply to such beamformers, and they need not be realizable by the sparser \ac{SC}-\ac{MiLAC}. The second question concerns both architectures. Most \ac{MiLAC} designs allow the susceptance of each tunable component, i.e., the imaginary part of its admittance, to take any real value, whereas practical tunable components, such as varactors or PIN diodes~\cite{Nerini2026Stem}, offer a limited tuning range and, under digital control, finitely many states. A learning-based design for lossy \acp{MiLAC} bounds each varactor capacitance~\cite{ZhouLossy}, and discrete component values have been evaluated only numerically~\cite{Nerini2025AnalogII}. Designs for \acp{RIS} with discrete phase shifts~\cite{Wu2020DiscreteRIS,Di2020HybridRIS}, phase quantization~\cite{Yang2021DiscreteRIS}, and phase-dependent reflection amplitudes~\cite{Abeywickrama2020Practical} do not apply directly either, since each \ac{RIS} element sets one diagonal entry of the reflection matrix, whereas each \ac{MiLAC} susceptance changes every entry of the beamformer through a matrix inversion.

\subsection{Contributions}

\leavevmode{}We study \ac{SC}-\ac{MiLAC}-aided \ac{MU-MISO} downlink beamforming without a digital beamformer. The main contributions are as follows.

\begin{itemize}
\item \textbf{Sum-Rate Optimality of the Stiefel Restriction:} The \ac{SC}-\ac{MiLAC} realizes every beamformer with orthonormal columns, i.e., every point of the complex Stiefel manifold, in closed form~\cite{Nerini2026Stem}. We prove that restricting the design to this manifold attains the optimal sum-rate of the \ac{FC}-\ac{MiLAC} whenever $N\ge 2K-1$, so the \ac{SC}-\ac{MiLAC} loses no sum-rate with $K(2N{+}1)$ instead of $(N{+}K)(N{+}K{+}1)/2$ tunable components (Proposition~\ref{prop:phase}). The proof maps every \ac{FC}-\ac{MiLAC} beamformer in closed form to a beamformer with orthonormal columns under which no user's \ac{SINR} decreases, so the result holds for every utility that is non-decreasing in each \ac{SINR}. It uses two properties: the \acp{SINR} depend on the beamformer only through its projection onto the user channels, and scaling this projection up never lowers the \acp{SINR}, which extends the result from $N\ge 2K$ to $N=2K-1$.
\item \textbf{Convergent Stiefel-Restricted Beamforming Algorithm:} For every $N\ge K$, we keep the beamformer on the Stiefel manifold within the \ac{WMMSE} framework by a Riemannian \ac{CG} update with a \ac{QR} retraction and an Armijo line search, so the \ac{SC}-\ac{MiLAC} realizes every iterate in closed form. We prove that every accumulation point of the iterates is stationary, for any accuracy of the inner loop (Proposition~\ref{prop:alg1_conv}).
\item \textbf{Beamforming with Bounded Discrete Susceptances:} For susceptances on a bounded uniform grid, set by its dynamic range and bit resolution, we propose a closed-form \ac{PHP} baseline and an \ac{AR} method for both the \ac{SC}-\ac{MiLAC} and the \ac{FC}-\ac{MiLAC}. \ac{AR} updates one susceptance at a time to maximize the sum-rate and thus never decreases it. It evaluates each candidate value by a rank-one update rather than a matrix inversion, which greatly reduces the computational cost of the refinement.
\item \textbf{Comprehensive Numerical Evaluation:} We compare the designs with fully digital and with \ac{PS}-based analog and hybrid beamforming. For $N\ge 2K-1$, the \ac{SC}-\ac{MiLAC} attains the \ac{FC}-\ac{MiLAC} sum-rate, and for $N\le 2K-2$, a numerical synthesis indicates that the remaining gap is caused by the Stiefel restriction rather than by the \ac{SC} topology. At $N=64$, the \ac{MiLAC} designs stay within $1.3$~bits/s/Hz of fully digital beamforming, perform close to hybrid beamforming, and exceed \ac{PS}-based analog beamforming by $5.7$ to $7.5$~bits/s/Hz. On a $3$-bit grid at a \ac{SNR} of $10$~dB, \ac{AR} attains $90\%$ (\ac{SC}) and $95\%$ (\ac{FC}) of the ideal sum-rate, lies between the two \ac{PS} baselines with $3$-bit \acp{PS}, and owes at least half of its loss to the starting point of its local search rather than to the grid. Relative to the estimated time with one matrix inversion per candidate, the rank-one updates reduce the runtime of each \ac{AR} sweep by factors of up to about $2500$ (\ac{SC}) and $5100$ (\ac{FC}).
\end{itemize}

\emph{Organization:} Section~\ref{sec:model} introduces the system model, Sections~\ref{sec:ideal} and~\ref{sec:synthesis} study the ideal and hardware-compliant cases, Section~\ref{sec:sim} presents simulations, and Section~\ref{sec:conc} concludes.

\emph{Notation:} Bold lowercase and uppercase letters denote vectors and matrices. $(\cdot)^\ast$, $(\cdot)^\TT$, $(\cdot)^\HH$, $(\cdot)^{-1}$, $\tr(\cdot)$, $\rank(\cdot)$, $\mathrm{Null}(\cdot)$, $\|\cdot\|$, $\|\cdot\|_2$, $\|\cdot\|_F$, $\mathbb{E}\{\cdot\}$, and $\realp{\cdot}$ denote conjugate, transpose, conjugate transpose, inverse, trace, rank, null space, Euclidean norm, spectral norm, Frobenius norm, expectation, and real part, and $j=\sqrt{-1}$. The singular values of $\Am$ are $\sigma_1(\Am)\ge\sigma_2(\Am)\ge\cdots$, and $\langle\Am,\Bm\rangle\triangleq\realp{\tr(\Am^\HH\Bm)}$. For Hermitian matrices, $\Am\preceq\Bm$ means that $\Bm-\Am$ is positive semidefinite, and $\Am^{1/2}$ is the positive-definite square root of a positive-definite $\Am$. The submatrix of $\Am$ with rows $a$ to $b$ and columns $c$ to $d$ is $[\Am]_{a:b,c:d}$, where a lone colon selects all rows or columns, and $[\xv]_{a:b}$ collects entries $a$ to $b$ of $\xv$. The diagonal matrix with diagonal $\xv$ is $\diag(\xv)$. The square root $\sqrt{\xv}$ and vector inequalities are taken entrywise, and $\mathbb{R}_+$ is the set of nonnegative real numbers. The $K\times K$ identity matrix, zero matrix or vector, and length-$K$ all-ones vector are $\Id_K$, $\Zerom$, and $\onev_K$. The complex circularly symmetric Gaussian with mean $\boldsymbol{\mu}$ and covariance $\boldsymbol{\Sigma}$ is $\mathcal{CN}(\boldsymbol{\mu},\boldsymbol{\Sigma})$, and the uniform distribution on $[a,b]$ is $\mathrm{Unif}(a,b)$.

\section{System Model}
\label{sec:model}

\subsection{MU-MISO Signal Model}

Consider a downlink \ac{MU-MISO} system in which an $N$-antenna \ac{BS} with $K\le N$ \ac{RF} chains serves $K$ single-antenna users. The data symbol vector $\sv\in\mathbb{C}^K$ satisfies $\mathbb{E}\{\sv\sv^\HH\}=\Id_K$, and a beamforming matrix $\Wm=[\wv_1,\ldots,\wv_K]\in\mathbb{C}^{N\times K}$ is applied, so the signal received at user $k$ is
\begin{equation}\label{eq:rx}
    r_k = \hv_k^\HH\wv_k s_k + \sum_{j\neq k}\hv_k^\HH\wv_j s_j + n_k,
\end{equation}
where $\hv_k\in\mathbb{C}^N$ is the channel from the \ac{BS} to user $k$, $\Hm\triangleq[\hv_1,\ldots,\hv_K]\in\mathbb{C}^{N\times K}$, and $n_k\sim\mathcal{CN}(0,\sigma^2)$ is the \ac{AWGN}. The \ac{SINR} and sum-rate are
\begin{equation}\label{eq:rate}
    \Gamma_k = \frac{|\hv_k^\HH\wv_k|^2}{\sum_{j\neq k}|\hv_k^\HH\wv_j|^2+\sigma^2},\
    \mathcal{R}(\Wm) = \sum_{k=1}^K\log_2(1+\Gamma_k).
\end{equation}
The fully digital sum-rate maximization under the total power budget $P_T$ is
\begin{equation}\label{eq:Pdig}
    \Wm^\star_{\mathrm{dig}} = \argmax_{\Wm:\tr(\Wm\Wm^\HH)\le P_T}\,\mathcal{R}(\Wm).
\end{equation}
Problem~\eqref{eq:Pdig} is non-convex but admits efficient stationary-point algorithms such as \ac{WMMSE}~\cite{Shi2011WMMSE}.

\begin{figure}[t]
\centering
\resizebox{\linewidth}{!}{%
\begin{tikzpicture}[
    >={Latex[length=1.6mm,width=1.3mm]},
    font=\footnotesize,
    box/.style={draw, rounded corners=1.5pt, align=center, inner sep=2pt},
    datab/.style={box, fill=blue!8, minimum width=11mm, minimum height=5mm},
    rfb/.style={box, fill=violet!14, minimum width=8.5mm, minimum height=5mm},
    milacb/.style={box, fill=red!10, minimum width=18mm, minimum height=20mm, font=\small},
    user/.style={box, fill=teal!14, minimum width=10.5mm, minimum height=5mm},
    addn/.style={draw, circle, inner sep=0.5pt, minimum size=3mm, font=\scriptsize},
    pab/.style={regular polygon, regular polygon sides=3, shape border rotate=-90, draw=blue!55!black, fill=blue!24, line width=0.6pt, inner sep=0pt, minimum size=3.6mm},
    link/.style={densely dashed, gray!65!black, line width=0.5pt},
    antenna/.pic={%
      \draw[line width=0.6pt] (0,0) -- (0,0.25);
      \draw[line width=0.6pt, line join=round, fill=orange!45] (0,0.25) -- (-0.15,0.5) -- (0.15,0.5) -- cycle;}
]
\useasboundingbox (-15.849pt,-42.331pt) rectangle (258.409pt,42.329pt);
\node[datab] (s1) at (0, 0.85) {$s_1$};
\node[datab] (sK) at (0, -0.85) {$s_K$};
\node[font=\scriptsize, gray] at (0, 0) {$\vdots$};
\node[rfb] (rf1) at (1.5, 0.85) {RF\,$1$};
\node[rfb] (rfK) at (1.5, -0.85) {RF\,$K$};
\node[font=\scriptsize, gray] at (1.5, 0) {$\vdots$};
\node[pab] (pa1) at (2.4, 0.85) {};
\node[pab] (paK) at (2.4, -0.85) {};
\node[font=\scriptsize, gray] at (2.4, 0) {$\vdots$};
\node[milacb] (milac) at (3.95, 0) {\ac{MiLAC}\\[2pt]$\Fm\in\mathbb{C}^{N\times K}$};
\coordinate (t1) at (5.35, 1.12); \coordinate (tN) at (5.35, -0.28);
\coordinate (r1) at (7.00, 1.12); \coordinate (rK) at (7.00, -0.28);
\draw[link] (t1) -- (r1); \draw[link] (tN) -- (r1);
\draw[link] (t1) -- (rK); \draw[link] (tN) -- (rK);
\node[font=\scriptsize] at (6.62, 1.33) {$\hv_1$};
\node[font=\scriptsize] at (6.62, -0.49) {$\hv_K$};
\draw (milac.east |- 0,0.7) -- (5.35, 0.7);
\draw (milac.east |- 0,-0.7) -- (5.35, -0.7);
\pic at (5.35, 0.7) {antenna};
\pic at (5.35, -0.7) {antenna};
\node[font=\scriptsize, gray] at (5.35, 0.22) {$\vdots$};
\pic at (7.00, 0.7) {antenna};
\pic at (7.00, -0.7) {antenna};
\node[addn] (n1) at (7.55, 0.7) {$+$};
\node[addn] (nK) at (7.55, -0.7) {$+$};
\node[font=\scriptsize, gray] (n1lab) at (7.55, 1.3) {$n_1$};
\node[font=\scriptsize, gray] (nKlab) at (7.55, -1.3) {$n_K$};
\draw[->, gray] (n1lab) -- (n1);
\draw[->, gray] (nKlab) -- (nK);
\draw[->] (7.00, 0.7) -- (n1.west);
\draw[->] (7.00, -0.7) -- (nK.west);
\node[user] (u1) at (8.55, 0.7) {User\,$1$};
\node[user] (uK) at (8.55, -0.7) {User\,$K$};
\node[font=\scriptsize, gray] at (8.55, 0) {$\vdots$};
\draw[->] (s1) -- (rf1);
\draw[->] (sK) -- (rfK);
\draw[->] (rf1.east) -- (pa1.west);
\draw[->] (rfK.east) -- (paK.west);
\draw[->] (pa1.east) -- (milac.west |- rf1)
  node[midway, above=1.5pt, font=\scriptsize] {$\sqrt{p_1}s_1$};
\draw[->] (paK.east) -- (milac.west |- rfK)
  node[midway, below=1.5pt, font=\scriptsize] {$\sqrt{p_K}s_K$};
\draw[->] (n1.east) -- (u1.west);
\draw[->] (nK.east) -- (uK.west);
\end{tikzpicture}%
}
\caption{\ac{MiLAC}-aided beamforming for the \ac{MU-MISO} downlink.}
\label{fig:sys}
\end{figure}

We consider an $(N{+}K)$-port lossless and reciprocal \ac{MiLAC} with scattering matrix $\Thetam\in\mathbb{C}^{(N+K)\times(N+K)}$, which satisfies $\Thetam^\HH\Thetam=\Id_{N+K}$ (lossless) and $\Thetam^\TT=\Thetam$ (reciprocal)~\cite{Nerini2025Capacity,Wu2026MiLAC}. With the $K$ input ports, connected to the \ac{RF} chains, numbered before the $N$ output ports, connected to the antennas, the analog beamformer is the off-diagonal block~\cite{Wu2026MiLAC}
\begin{equation}\label{eq:F_block}
    \Fm\triangleq[\Thetam]_{K+1:K+N,\,1:K}\in\mathbb{C}^{N\times K},\qquad \|\Fm\|_2\le 1,
\end{equation}
where the spectral bound follows from the unitarity of $\Thetam$.\footnote{Following the convention of~\cite{Wu2026MiLAC}, $\Fm$ in~\eqref{eq:F_block} equals $2 \Fm_{\mathrm{MiLAC}}$, where the physical analog beamformer is $\Fm_{\mathrm{MiLAC}} = (1/2)[\Thetam]_{K+1:K+N,\,1:K}$ as in~\cite{Nerini2025Capacity}. The factor $1/2$ is the voltage division between the source impedance of each amplifier and the matched port it drives~\cite{Nerini2025Capacity}. It applies equally to the other schemes, e.g., the fully digital benchmark~\cite{Wu2026MiLAC},~\cite{Nerini2025Capacity}, so it is absorbed into the \ac{SNR} and does not affect any comparison in this paper.} The scattering matrix follows from a real symmetric susceptance matrix $\Bm\in\mathbb{R}^{(N+K)\times(N+K)}$ through the Cayley transform~\cite{Nerini2025Capacity}
\begin{equation}\label{eq:cayley}
\Thetam=(Y_0\Id-j\Bm)(Y_0\Id+j\Bm)^{-1},
\end{equation}
where $Y_0=1/Z_0$ is the reference admittance with $Z_0$ the port reference impedance. Combining~\eqref{eq:F_block} and~\eqref{eq:cayley}, the beamformer depends on the susceptances only through the normalized susceptances $\Bm/Y_0$.

Figure~\ref{fig:sys} depicts the signal flow with a \ac{MiLAC}: no beamforming is applied at baseband, and the power of each stream is set by an amplifier (blue triangle) at the output of its \ac{RF} chain, so the overall beamformer factorizes as
\begin{equation}\label{eq:WfromF}
    \Wm = \Fm\diag(\sqrt{\pv}),\quad \Fm\in\mathbb{C}^{N\times K},\quad \pv\in\mathbb{R}_+^K,
\end{equation}
with $\Fm$ the analog beamformer in~\eqref{eq:F_block} and $\pv$ the per-stream powers constrained by $\onev_K^\TT\pv\le P_T$. This placement follows passive phased arrays, in which high-power amplification precedes a low-loss feed network~\cite{Herd2016PhasedArray,LopezValcarce2022Hybrid}. Hence $p_k$ is the power that the amplifier of stream $k$ injects into input port $k$ of the \ac{MiLAC}, and the budget bounds the total amplifier output power. Since the \ac{MiLAC} is lossless, the injected power is either delivered to the antennas through the output ports or reflected back to the input ports. With uncorrelated unit-power symbols,
\begin{equation}\label{eq:power_split}
\underbrace{\onev_K^\TT\pv}_{\text{injected}}=\underbrace{\tr(\Wm\Wm^\HH)}_{\text{radiated}}+\underbrace{\tr\bigl(\diag(\pv)(\Id_K-\Fm^\HH\Fm)\bigr)}_{\text{reflected}\,\ge\,0}.
\end{equation}
The radiated power therefore never exceeds $P_T$, and it equals the injected power when $\Fm$ has orthonormal columns (${\Fm^\HH\Fm=\Id_K}$), as for the \ac{SC}-\ac{MiLAC} designs of Section~\ref{sec:ideal}. The model~\eqref{eq:WfromF} is the pure-analog setting that~\cite{Wu2026MiLAC} calls \emph{\ac{MiLAC}-aided beamforming}, as opposed to \emph{hybrid digital-\ac{MiLAC} beamforming}, which inserts a non-diagonal digital beamformer. We explain this choice in Remark~\ref{rem:hybrid}.

\begin{figure}[t]
\centering
\begin{tikzpicture}[
    font=\scriptsize,
    rf/.style={draw, fill=blue!18, rectangle, minimum size=4mm, inner sep=0pt, font=\scriptsize},
    rfc/.style={draw, line width=0.7pt, fill=blue!28, rectangle, minimum size=4mm, inner sep=0pt, font=\scriptsize},
    ant/.style={draw, fill=orange!50, regular polygon, regular polygon sides=3, inner sep=0pt, minimum size=3.2mm},
    antc/.style={draw, line width=0.7pt, fill=orange!85!yellow, regular polygon, regular polygon sides=3, inner sep=0pt, minimum size=3.5mm},
    susc/.style={red!70!black, line width=0.5pt},
    cgframe/.style={densely dotted, gray!75!black, line width=0.5pt, rounded corners=2pt}
]
\begin{scope}[xshift=0cm]
\node[rf] (f1) at (0,   1.1) {};
\node[rf] (f2) at (1.8, 1.1) {};
\foreach \i/\x in {1/-0.15, 2/0.55, 3/1.25, 4/1.95} {
  \node[ant] (fa\i) at (\x, -0.1) {};
}
\draw[susc] (f1) to[bend left=28] (f2);
\draw[susc] (fa1) -- (fa2);
\draw[susc] (fa2) -- (fa3);
\draw[susc] (fa3) -- (fa4);
\draw[susc] (fa1) to[bend right=20] (fa3);
\draw[susc] (fa2) to[bend right=20] (fa4);
\draw[susc] (fa1) to[bend right=28] (fa4);
\foreach \u in {f1,f2} {
  \foreach \v in {fa1,fa2,fa3,fa4} {
    \draw[susc] (\u) -- (\v);
  }
}
\node[font=\scriptsize] at (0.9, -0.7) {(a)};
\end{scope}
\begin{scope}[xshift=4.6cm]
\node[rfc]  (s1) at (0,    1.1) {};
\node[antc] (sc) at (0.9,  1.1) {};
\node[rfc]  (s2) at (1.8,  1.1) {};
\node[ant] (sa1) at (0.15, -0.1) {};
\node[ant] (sa2) at (0.9,  -0.1) {};
\node[ant] (sa3) at (1.65, -0.1) {};
\draw[susc] (s1) -- (sc);
\draw[susc] (sc) -- (s2);
\draw[susc] (s1) to[bend left=38] (s2);
\foreach \a in {sa1,sa2,sa3} {
  \foreach \c in {s1,sc,s2} {
    \draw[susc] (\a) -- (\c);
  }
}
\draw[cgframe] (-0.32, 0.83) rectangle (2.12, 1.62);
\node[font=\tiny, gray!75!black, anchor=south] at (0.9, 1.55) {central vertices};
\node[font=\scriptsize] at (0.9, -0.7) {(b)};
\end{scope}
\node[draw, fill=blue!18, rectangle, minimum size=2.5mm, inner sep=0pt] (legrf) at (0, 2) {};
\node[font=\scriptsize, anchor=west] at (0.2, 2) {port connected to \ac{RF} chain};
\node[regular polygon, regular polygon sides=3, draw, fill=orange!50, inner sep=0pt, minimum size=2.5mm] (legant) at (3.65, 2) {};
\node[font=\scriptsize, anchor=west] at (3.85, 2) {port connected to antenna};
\end{tikzpicture}
\vspace{-2mm}
\caption{\ac{MiLAC} architectures: (a) \ac{FC}-\ac{MiLAC}, (b) \ac{SC}-\ac{MiLAC}.}
\label{fig:topology}
\end{figure}

\subsection{MiLAC Architecture}
\label{sec:milac_arch}
We model a \ac{MiLAC} as a graph whose $N{+}K$ vertices are its ports and whose edges are the tunable components that connect pairs of ports~\cite{Nerini2026Stem}. We consider two architectures, which differ in their graphs and hence in their numbers of tunable components.

\textbf{\ac{FC}-\ac{MiLAC}:} The graph is complete. With one tunable component between each pair of ports and one from each port to ground, the number of tunable components is
\begin{equation}\label{eq:Nfull}
    N_C^{\mathrm{full}} = (N{+}K)(N{+}K{+}1)/2.
\end{equation}

\textbf{\ac{SC}-\ac{MiLAC}:}  The graph has $Q=2K{-}1$ central vertices, each connected to every other vertex, while the non-central vertices connect only to the central ones~\cite{Nerini2026Stem}, where $Q$ is called the center size. The $K$ input ports are central, so $K{-}1$ output ports are also central, and the remaining $N{-}K{+}1$ output ports are non-central. The central ports are connected pairwise through $Q(Q-1)/2=(2K-1)(K-1)$ components, the $N-K+1$ non-central ports are connected to the $Q$ central ports through $(2K-1)(N-K+1)$ components, and each of the $N+K$ ports has one component to ground. Summing these numbers gives the number of tunable components~\cite{Nerini2026Stem}
\begin{equation}\label{eq:Nstem}
    N_C^{\mathrm{stem}} = K(2N+1),
\end{equation}
which is linear in $N$ for fixed $K$.

Figure~\ref{fig:topology} shows both topologies for $(N,K)=(4,2)$, with squares for the \ac{RF}-chain ports and triangles for the antenna ports. In the \ac{SC} panel, the $Q=3$ central vertices (heavier outlines) connect to one another and to each of the $N{-}K{+}1=3$ non-central ports. Port-to-ground components are not drawn.

The tunable components of a \ac{MiLAC} are the admittances $Y_{i,k}=jB_{i,k}$ between connected ports $i\ne k$ and the admittances $Y_{k,k}=jB_{k,k}$ from each port $k$ to ground, where the susceptances $B_{i,k}$ are real since the \ac{MiLAC} is lossless, and $B_{i,k}=0$ for unconnected pairs. The entries of $\Bm$ are linear combinations of these susceptances~\cite{Nerini2025AnalogI},
\begin{equation}\label{eq:Bmat_from_phys}
[\Bm]_{i,k} =
\begin{cases}
\displaystyle -B_{i,k}, & i\ne k,\\[2pt]
\displaystyle \sum_{m=1}^{N+K} B_{m,k}, & i=k,
\end{cases}
\end{equation}
and, conversely, $B_{i,k}=-[\Bm]_{i,k}$ for $i\ne k$ and $B_{k,k}=\sum_{m=1}^{N+K}[\Bm]_{m,k}$.  The entry $[\Bm]_{k,k}$ sums the susceptances of all components at port $k$, so it is not itself a tunable component.  We write $\bv\in\mathbb{R}^{N_C}$ for the vector that stacks the susceptances of the $N_C$ tunable components in a fixed order, with $N_C\in\{N_C^{\mathrm{stem}},N_C^{\mathrm{full}}\}$.  The analog beamformer is then $\Fm=\Phi(\bv)\triangleq[\Thetam(\bv)]_{K+1:K+N,\,1:K}$, where $\Thetam(\bv)$ follows from~\eqref{eq:cayley} with $\Bm$ assembled from $\bv$ via~\eqref{eq:Bmat_from_phys}, and $\Phi$ maps $\mathbb{R}^{N_C}$ into the unit spectral-norm ball $\{\Fm\in\mathbb{C}^{N\times K}:\|\Fm\|_2\le 1\}$ by~\eqref{eq:F_block}. The map $\Phi$ depends on the architecture through the components that $\bv$ contains, and for the \ac{SC}-\ac{MiLAC}, the set of realizable analog beamformers is $\mathcal{F}_{\mathrm{stem}}\triangleq\Phi(\mathbb{R}^{N_C^{\mathrm{stem}}})$.  Bounding and quantizing $\bv$ entrywise (Section~\ref{sec:hw_model}) is then physically meaningful: each entry of $\bv$ corresponds to one component (e.g., a varactor) with its own dynamic range and bias-voltage resolution.

The inverse task, which turns a designed beamformer $\Fm$ into susceptances, is called synthesis and denoted $\bv=\Psi(\Fm)$. For the \ac{SC}-\ac{MiLAC} and $\Fm$ with orthonormal columns,~\cite{Nerini2026Stem} requires the first $K$ columns of $\Thetam$ to equal $[\Zerom,\Fm^\TT]^\TT$, i.e., no reflection at the input ports and $\Fm$ at the output ports. By~\eqref{eq:cayley}, this is a set of linear equations in $\Bm$, and~\cite[Alg.~1]{Nerini2026Stem} solves them in closed form for the \ac{SC} sparsity pattern: first the diagonal and off-diagonal parts of the antenna-side block of $\Bm$, through ${N-K+1}$ linear systems of size $K$, and then the remaining blocks. For the \ac{FC}-\ac{MiLAC} and any $\Fm$ with $\|\Fm\|_2\le 1$,~\cite[Prop.~1]{Wu2026MiLAC} completes $\Fm$ to a symmetric unitary $\Thetam$, and $\Bm$ follows from the inverse Cayley transform $j\Bm/Y_0=(\Id-\Thetam)(\Id+\Thetam)^{-1}$. In both cases, the susceptances of the individual components then follow from $\Bm$ by inverting~\eqref{eq:Bmat_from_phys}, and setting each component to its value realizes $\Fm$.

\subsection{Two Base Results}
\label{sec:base}

We restate two known results on which our analysis builds.

\begin{proposition}[Stiefel realizability via \ac{SC}-\ac{MiLAC}, {\cite[Prop.~1]{Nerini2026Stem}}]
\label{prop:stem}
For an \ac{SC}-\ac{MiLAC} with center size $Q=2K{-}1$ whose central vertices include the $K$ \ac{RF}-chain ports, and for any $\Fm^\star$ in the complex Stiefel manifold $\St(N,K)\triangleq{\{\Fm\in\mathbb{C}^{N\times K}:\Fm^\HH\Fm=\Id_K\}}$, there exists $\bv^\star\in\mathbb{R}^{N_C^{\mathrm{stem}}}$ with $\Phi(\bv^\star)=\Fm^\star$, i.e., $\St(N,K)\subseteq\mathcal{F}_{\mathrm{stem}}$, and the closed-form synthesis $\bv^\star=\Psi(\Fm^\star)$ of~\cite[Alg.~1]{Nerini2026Stem} returns such a $\bv^\star$.
\end{proposition}

\begin{remark}[Scope of Proposition~\ref{prop:stem}]
\label{rem:stiefel_suff}
Proposition~\ref{prop:stem} and~\eqref{eq:F_block} give only the inclusions $\St(N,K)\subseteq\mathcal{F}_{\mathrm{stem}}\subseteq{\{\Fm:\|\Fm\|_2\le 1\}}$. The first inclusion is strict, since $\bv=\zerov$ gives $\Thetam=\Id$ and $\Fm=\Zerom$. Whether $\mathcal{F}_{\mathrm{stem}}$ equals the spectral-norm ball is open. The \ac{SC}-\ac{MiLAC} has $K(2N{+}1)$ tunable components, $K$ more than the $2NK$ real parameters of $\Fm$, so the number of unknowns does not rule out equality. Since the closed-form synthesis of~\cite{Nerini2026Stem} applies only to beamformers with orthonormal columns, we henceforth restrict the design to $\Fm\in\St(N,K)$.
\end{remark}

\begin{proposition}[Hybrid digital-\ac{FC}-\ac{MiLAC} achievability, {\cite[Prop.~3]{Wu2026MiLAC}}]
\label{prop:wu}
Consider the hybrid digital-\ac{FC}-\ac{MiLAC} architecture with $K$ \ac{RF} chains and beamformer $\Wm=\Fm\Pm$, where the analog beamformer satisfies $\|\Fm\|_2\le 1$~\cite{Wu2026MiLAC} and $\Pm\in\mathbb{C}^{K\times K}$ is a digital beamformer. For any $\Wm\in\mathbb{C}^{N\times K}$ with $\tr(\Wm\Wm^\HH)\le P_T$, the \ac{SVD} $\Wm=\Um_W\Sigmam_W\Vm_W^\HH$ gives $\Fm=[\Um_W]_{:,1:K}$ and $\Pm=[\Sigmam_W]_{1:K,:}\Vm_W^\HH$, so the hybrid architecture attains the fully digital sum-rate~\eqref{eq:Pdig}.
\end{proposition}

\begin{remark}[Hybrid digital-\ac{SC}-\ac{MiLAC} beamforming]
\label{rem:hybrid}
The factor $\Fm=[\Um_W]_{:,1:K}$ in Proposition~\ref{prop:wu} has orthonormal columns, so an \ac{SC}-\ac{MiLAC} with $K(2N{+}1)$ tunable components realizes it by Proposition~\ref{prop:stem}, and a hybrid digital-\ac{SC}-\ac{MiLAC} architecture attains the fully digital sum-rate. This follows directly from Propositions~\ref{prop:stem} and~\ref{prop:wu}. We focus instead on the pure-analog setting~\eqref{eq:WfromF}, which needs no symbol-rate baseband processing (Section~\ref{sec:intro}) and in which the multiuser performance of the \ac{SC}-\ac{MiLAC} is open~\cite{Nerini2026Stem}. Without a digital beamformer, $\Fm$ alone must control the inter-user interference, the problem addressed next.
\end{remark}

\section{Stem-Connected MiLAC-Aided MU-MISO Beamforming: Ideal Case}
\label{sec:ideal}

This section designs \ac{SC}-\ac{MiLAC}-aided beamforming under the ideal assumption that every susceptance can take any real value, i.e., $\bv\in\mathbb{R}^{N_C^{\mathrm{stem}}}$.

\subsection{Problem Formulation on the Stiefel Manifold}
\label{sec:ideal_form}

For the \ac{SC}-\ac{MiLAC}, the beamformer in~\eqref{eq:rate} must satisfy
\begin{equation}\label{eq:pure_model}
    \Wm = \Fm\,\diag(\sqrt{\pv}),\ \Fm\in\mathcal{F}_{\mathrm{stem}},\ \onev_K^\TT\pv\le P_T,
\end{equation}
where $\mathcal{F}_{\mathrm{stem}}$ is the set of \ac{SC}-\ac{MiLAC}-realizable analog beamformers defined in Section~\ref{sec:milac_arch}. Optimizing directly over $\mathcal{F}_{\mathrm{stem}}$ faces two obstacles: the set has no closed-form description, and deciding membership requires solving the nonlinear equation $\Phi(\bv)=\Fm$ for $\bv$. We therefore restrict the design to beamformers with orthonormal columns, i.e., to $\Fm\in\St(N,K)$, which the \ac{SC}-\ac{MiLAC} realizes in closed form by Proposition~\ref{prop:stem}. We call this the Stiefel restriction and formulate
\begin{equation}\label{eq:Pstem}
\mathcal{P}_{\mathrm{stem}}:\;\max_{\Fm\in\St(N,K),\,\pv\in\mathbb{R}_+^K,\,\onev_K^\TT\pv\le P_T}\;\mathcal{R}\bigl(\Fm\diag(\sqrt{\pv})\bigr).
\end{equation}
Proposition~\ref{prop:phase} in Section~\ref{sec:phase} shows that this restriction loses no sum-rate when $N\ge 2K-1$.

\subsection{Sum-Rate Optimality of the Stiefel Restriction}
\label{sec:phase}

We compare~\eqref{eq:Pstem} with its \ac{FC}-\ac{MiLAC}-aided counterpart~\cite{Wu2026MiLAC}, which replaces the Stiefel constraint by the spectral-norm ball,
\begin{equation}\label{eq:Pfull}
\mathcal{P}_{\mathrm{full}}:\;\max_{\Fm:\|\Fm\|_2\le 1,\,\pv\in\mathbb{R}_+^K,\,\onev_K^\TT\pv\le P_T}\;\mathcal{R}\bigl(\Fm\diag(\sqrt{\pv})\bigr).
\end{equation}
Let $\mathcal{R}_{\mathrm{stem}}^\star(\Hm)$ and $\mathcal{R}_{\mathrm{full}}^\star(\Hm)$ denote the optimal values of~\eqref{eq:Pstem} and~\eqref{eq:Pfull}, respectively, where $\Hm$ has full column rank, which holds almost surely for continuously distributed channels such as Rayleigh and Rician fading.

\begin{proposition}[Sum-rate of the Stiefel restriction]
\label{prop:phase}
Let ${K\le N}$. Then $\mathcal{R}_{\mathrm{stem}}^\star(\Hm)\le\mathcal{R}_{\mathrm{full}}^\star(\Hm)$ for every $\Hm$, with equality whenever $N\ge 2K-1$. Moreover, for $N\ge 2K-1$, the construction in the proof maps every feasible point of~\eqref{eq:Pfull} to a feasible point of~\eqref{eq:Pstem} with the same power allocation and no smaller sum-rate.
\end{proposition}

\begin{proof}
\emph{Step~1 (common coordinates):} Let the columns of $\Nm_H\in\mathbb{C}^{N\times(N-K)}$ form an orthonormal basis of $\mathrm{Null}(\Hm^\HH)$, and let $\overline{\Hm}\triangleq\Hm^\HH\Hm$, which is positive definite. Since $[\Hm\overline{\Hm}^{-1/2},\Nm_H]$ is unitary, every $\Fm\in\mathbb{C}^{N\times K}$ decomposes uniquely as
\begin{equation}\label{eq:FhF}
\Fm=\Hm\overline{\Hm}^{-1/2}\Ym+\Nm_H\Vm,\quad \Ym\triangleq\overline{\Hm}^{-1/2}\Hm^\HH\Fm,\quad \Vm\triangleq\Nm_H^\HH\Fm,
\end{equation}
with $\Ym\in\mathbb{C}^{K\times K}$, $\Vm\in\mathbb{C}^{(N-K)\times K}$, and
\begin{equation}\label{eq:FFalg}
\Fm^\HH\Fm=\Ym^\HH\Ym+\Vm^\HH\Vm.
\end{equation}
Since $\Hm^\HH\Fm=\overline{\Hm}^{1/2}\Ym$, the sum-rate depends on $\Fm$ only through $\Ym$, and we write it as $\mathcal{R}(\Ym,\pv)$.

\emph{Step~2 (feasible sets):} By~\eqref{eq:FFalg}, $\|\Fm\|_2\le 1$, i.e., $\Fm^\HH\Fm\preceq\Id_K$, holds if and only if $\Ym^\HH\Ym+\Vm^\HH\Vm\preceq\Id_K$. Since $\Vm^\HH\Vm\succeq\Zerom$, and setting $\Vm=\Zerom$ preserves feasibility~\cite{Wu2026MiLAC}, the set of $\Ym$ attainable in~\eqref{eq:Pfull} is
\begin{equation}\label{eq:Yfull_set}
\mathcal{Y}_{\mathrm{full}}=\{\Ym\in\mathbb{C}^{K\times K}:\Ym^\HH\Ym\preceq\Id_K\},
\end{equation}
which recovers the reformulation in~\cite{Wu2026MiLAC}. By~\eqref{eq:FFalg}, $\Fm^\HH\Fm=\Id_K$ holds if and only if $\Vm^\HH\Vm=\Id_K-\Ym^\HH\Ym$, which has a solution $\Vm$ with $N-K$ rows if and only if $\Id_K-\Ym^\HH\Ym$ is positive semidefinite with rank at most $N-K$. Hence the set of $\Ym$ attainable in~\eqref{eq:Pstem} is
\begin{equation}\label{eq:Ystem_set}
\mathcal{Y}_{\mathrm{stem}}=\{\Ym\in\mathcal{Y}_{\mathrm{full}}:\rank(\Id_K-\Ym^\HH\Ym)\le N-K\},
\end{equation}
and $\mathcal{Y}_{\mathrm{stem}}\subseteq\mathcal{Y}_{\mathrm{full}}$ gives $\mathcal{R}_{\mathrm{stem}}^\star(\Hm)\le\mathcal{R}_{\mathrm{full}}^\star(\Hm)$.

\emph{Step~3 (equality for $N\ge 2K-1$):} Let $(\Fm,\pv)$ be any feasible point of~\eqref{eq:Pfull}, with $\Ym$ given by~\eqref{eq:FhF}. If $\Ym=\Zerom$, the sum-rate is zero and no point of~\eqref{eq:Pstem} does worse. Otherwise, scaling $\Ym$ by $t\ge 1$ multiplies the signal and interference powers in every \ac{SINR} in~\eqref{eq:rate} by $t^2$ and leaves $\sigma^2$ unchanged, so the sum-rate does not decrease. Hence $\widetilde{\Ym}\triangleq\Ym/\sigma_1(\Ym)$ satisfies $\sigma_1(\widetilde{\Ym})=1$, so $\widetilde{\Ym}\in\mathcal{Y}_{\mathrm{full}}$, and $\mathcal{R}(\widetilde{\Ym},\pv)\ge\mathcal{R}(\Ym,\pv)$, since $\Ym\in\mathcal{Y}_{\mathrm{full}}$ gives $1/\sigma_1(\Ym)\ge 1$. The eigenvalues of $\Id_K-\widetilde{\Ym}^\HH\widetilde{\Ym}$ are $1-\sigma_k(\widetilde{\Ym})^2$, and the one for $k=1$ is zero, so $\rank(\Id_K-\widetilde{\Ym}^\HH\widetilde{\Ym})\le K-1\le N-K$ and $\widetilde{\Ym}\in\mathcal{Y}_{\mathrm{stem}}$. Choosing $\widetilde{\Vm}$ with $\widetilde{\Vm}^\HH\widetilde{\Vm}=\Id_K-\widetilde{\Ym}^\HH\widetilde{\Ym}$ yields $\widetilde{\Fm}\triangleq\Hm\overline{\Hm}^{-1/2}\widetilde{\Ym}+\Nm_H\widetilde{\Vm}\in\St(N,K)$, which is feasible for~\eqref{eq:Pstem} with the same $\pv$ and has a sum-rate no smaller than that of $(\Fm,\pv)$. Applying this to a maximizer of~\eqref{eq:Pfull}, which exists since its feasible set is compact, gives $\mathcal{R}_{\mathrm{stem}}^\star(\Hm)\ge\mathcal{R}_{\mathrm{full}}^\star(\Hm)$. For $N\ge 2K$, $\mathcal{Y}_{\mathrm{stem}}=\mathcal{Y}_{\mathrm{full}}$, since $K\le N-K$.
\end{proof}

\begin{remark}[Relation between Propositions~\ref{prop:stem} and~\ref{prop:phase}]
\label{rem:insight}
Proposition~\ref{prop:stem} is a realizability result: the \ac{SC}-\ac{MiLAC} realizes every beamformer in $\St(N,K)$. It does not show that~\eqref{eq:Pfull} has an optimal solution in $\St(N,K)$. Proposition~\ref{prop:phase} fills this gap by showing that such a solution exists whenever $N\ge 2K-1$, so the Stiefel restriction loses no sum-rate in this case. This conclusion does not follow by applying Proposition~\ref{prop:stem} directly to the solutions of~\eqref{eq:Pfull} computed by~\cite[Alg.~2]{Wu2026MiLAC}. These solutions lie in the range of $\Hm$ and hence have $\Vm=\Zerom$ in~\eqref{eq:FhF}, so by~\eqref{eq:FFalg} they have orthonormal columns only if $\Ym$ is unitary.

The proof of Proposition~\ref{prop:phase} relies on two properties of the \ac{MU-MISO} sum-rate. First, it depends on $\Fm$ only through the $K\times K$ matrix $\Ym$, so the component $\Nm_H\Vm$ in the null space of $\Hm^\HH$ can be chosen freely to complete $\Fm$ to orthonormal columns, which is possible whenever $\rank(\Id_K-\Ym^\HH\Ym)\le N-K$. Second, scaling $\Ym$ up never lowers the sum-rate, so $\Ym$ can be scaled to $\sigma_1(\Ym)=1$, which makes this rank at most $K-1$ and lowers the required number of antennas from $2K$ to $2K-1$. Since the construction lowers no \ac{SINR}, the proof holds for every utility that is non-decreasing in each \ac{SINR}, e.g., the weighted sum-rate or the minimum \ac{SINR}. For $N\ge 2K-1$, the construction also maps the solutions of~\cite[Alg.~2]{Wu2026MiLAC} to $\St(N,K)$, whereas Algorithm~\ref{alg:stem} of Section~\ref{sec:stem_overall} solves~\eqref{eq:Pstem} directly for every $N\ge K$, including $N\le 2K-2$.
\end{remark}

For $K\ll N$, the ratio
\begin{equation}\label{eq:ratio}
    \frac{N_C^{\mathrm{full}}}{N_C^{\mathrm{stem}}} = \frac{(N{+}K)(N{+}K{+}1)}{2K(2N{+}1)} \;\approx\; \frac{N}{4K}
\end{equation}
grows linearly with $N$. For example, when $K=4$, the \ac{SC}-\ac{MiLAC} needs $4.5$, $8.5$, and $16.5$ times fewer tunable components than the \ac{FC}-\ac{MiLAC} at $N=64$, $128$, and $256$.

\begin{remark}[Small arrays, $K\le N\le 2K-2$]
\label{rem:binding}
For $N<2K$, the rank condition in~\eqref{eq:Ystem_set} requires at least $2K-N$ singular values of $\Ym$ to equal one, whereas $\mathcal{Y}_{\mathrm{full}}$ imposes no such condition. For $N=2K-1$, one singular value must equal one, and the scaling in Step~3 of the proof of Proposition~\ref{prop:phase} meets this condition without loss. For $N\le 2K-2$, Proposition~\ref{prop:phase} leaves open whether $\mathcal{R}_{\mathrm{stem}}^\star(\Hm)<\mathcal{R}_{\mathrm{full}}^\star(\Hm)$ can occur.\footnote{For \ac{BD-RIS}-aided multiuser \ac{MIMO}, reduced-complexity architectures, including a stem-connected one with a large enough stem, match the performance of the fully-connected architecture~\cite{Wu2025BDRIS}.} Such a loss would be caused by the restriction of~\eqref{eq:Pstem} to $\St(N,K)$, not necessarily by the \ac{SC} topology, since $\mathcal{F}_{\mathrm{stem}}$ may contain maximizers of~\eqref{eq:Pfull} outside $\St(N,K)$ (Remark~\ref{rem:stiefel_suff}). Section~\ref{sec:sim_ideal} examines this case numerically.
\end{remark}

\subsection{WMMSE-Based Solution}
\label{sec:wmmse}

We solve~\eqref{eq:Pstem} with the \ac{WMMSE} method of~\cite{Shi2011WMMSE}. Let $u_k\in\mathbb{C}$ denote the scalar receive filter of user $k$. The mean-square error of user $k$ with the beamformer $\wv_k=\sqrt{p_k}\fv_k$, where $\fv_k$ is the $k$th column of $\Fm$, is
\begin{align}\label{eq:mse_stem}
    e_k(\Fm,\pv) = &\bigl|1-u_k^\ast\hv_k^\HH\fv_k\sqrt{p_k}\bigr|^2 + \sum_{j\ne k}p_j\bigl|u_k^\ast\hv_k^\HH\fv_j\bigr|^2\notag \\
    &+ |u_k|^2\sigma^2.
\end{align}
With weights $\omega_k>0$, we collect the $u_k$ and $\omega_k$ into $\uv=[u_1,\ldots,u_K]^\TT$ and $\omegav=[\omega_1,\ldots,\omega_K]^\TT$. The \ac{WMMSE} identity~\cite{Shi2011WMMSE},
\begin{equation}\label{eq:wmmse_identity_stem}
    \min_{u_k,\omega_k>0}\bigl(\omega_k e_k(\Fm,\pv)-\ln\omega_k\bigr) = 1 - \ln\bigl(1+\Gamma_k(\Fm,\pv)\bigr),
\end{equation}
transforms~\eqref{eq:Pstem} equivalently into
\begin{equation}\label{eq:wmmse_stem}
    \min_{\Fm\in\St(N,K),\,\pv\ge\zerov,\,\onev_K^\TT\pv\le P_T,\,\uv,\,\omegav}\ \sum_{k=1}^K\bigl(\omega_k e_k(\Fm,\pv)-\ln\omega_k\bigr).
\end{equation}
We solve~\eqref{eq:wmmse_stem} by alternating minimization over the four variable blocks $(\uv,\omegav,\pv,\Fm)$.

\subsubsection{Update of $u_k$ and $\omega_k$}
\label{sec:stem_u_omega}

For fixed $(\Fm,\pv)$, the problem decouples across users. Minimizing~\eqref{eq:mse_stem} over $u_k$, and then the objective of~\eqref{eq:wmmse_stem} over $\omega_k>0$, gives the \ac{MMSE} receiver and the weight
\begin{equation}\label{eq:u_update_stem}
    u_k = \frac{\sqrt{p_k}\,\hv_k^\HH\fv_k}{\sum_{j=1}^K p_j|\hv_k^\HH\fv_j|^2 + \sigma^2},\quad \omega_k = \frac{1}{e_k(\Fm,\pv)},
\end{equation}
where $e_k(\Fm,\pv)$ is evaluated at the updated $u_k$.

\subsubsection{Update of $\pv$}
\label{sec:stem_p_update}

For fixed $(\Fm,\uv,\omegav)$, minimizing the objective of~\eqref{eq:wmmse_stem} over $\pv$ reduces, up to a constant, to $\min_{\pv\ge\zerov,\,\onev_K^\TT\pv\le P_T}\sum_{k=1}^K(\eta_kp_k-2\nu_k\sqrt{p_k})$ with
\begin{equation}\label{eq:a_stem}
    \nu_k = \omega_k\realp{u_k^\ast\hv_k^\HH\fv_k},\quad \eta_k = \sum_{i=1}^K\omega_i|u_i|^2|\hv_i^\HH\fv_k|^2.
\end{equation}
This subproblem is convex in $\pv$, since $\nu_k\ge 0$ after~\eqref{eq:u_update_stem} and $\sqrt{p_k}$ is concave. Its \ac{KKT} conditions give
\begin{equation}\label{eq:p_update_stem}
    p_k = \left[\frac{\nu_k}{\eta_k+\mu_p}\right]_+^2,
\end{equation}
where $[x]_+\triangleq\max\{x,0\}$ and $\mu_p\ge 0$ is the dual variable of the budget $\onev_K^\TT\pv\le P_T$, found by bisection.

\subsubsection{Update of $\Fm$ on the Stiefel Manifold}
\label{sec:stem_F_update}

For fixed $(\pv,\uv,\omegav)$, the $\Fm$-subproblem is
\begin{equation}\label{eq:F_subprob_stem}
    \min_{\Fm\in\St(N,K)}\ \mathcal{J}(\Fm;\pv,\Um,\Omegam),
\end{equation}
where, with $\Um=\diag(\uv)$ and $\Omegam=\diag(\omegav)$, the matrix $\Em \triangleq \Id_K-\Um^\HH\Hm^\HH\Fm\diag(\sqrt{\pv})$ satisfies $e_k(\Fm,\pv)={[\Em\Em^\HH+\sigma^2\Um^\HH\Um]_{k,k}}$, so the objective of~\eqref{eq:wmmse_stem} equals
\begin{equation}\label{eq:Jsurr_stem}
    \mathcal{J}(\Fm;\pv,\Um,\Omegam) = \tr\!\bigl(\Omegam[\Em\Em^\HH+\sigma^2\Um^\HH\Um]\bigr)-\ln\det\Omegam.
\end{equation}

To solve~\eqref{eq:F_subprob_stem}, we use the fact that $\St(N,K)$ is a compact embedded Riemannian submanifold of $\mathbb{C}^{N\times K}$~\cite{Edelman1998Geometry,Absil2008Optimization}. Its tangent space at $\Fm\in\St(N,K)$ is $T_\Fm\St(N,K) = {\{\Xim\in\mathbb{C}^{N\times K}:\Fm^\HH\Xim+\Xim^\HH\Fm=\Zerom\}}$, and the orthogonal projection of $\Xm\in\mathbb{C}^{N\times K}$ onto it is
\begin{equation}\label{eq:Stiefel_proj}
    \mathrm{P}_\Fm(\Xm) = \Xm-\frac{1}{2}\Fm\,\left(\Fm^\HH\Xm+\Xm^\HH\Fm\right),
\end{equation}
and the \ac{QR} retraction $R_\Fm(\Xim) = \mathrm{qf}(\Fm+\Xim)$ maps a tangent vector $\Xim$ back onto $\St(N,K)$, where $\mathrm{qf}(\Am)$ denotes the factor with orthonormal columns in the thin \ac{QR} factorization of $\Am$, whose triangular factor has a positive diagonal~\cite{Absil2008Optimization}. The Euclidean gradient of~\eqref{eq:Jsurr_stem} with respect to $\Fm$ under $\langle\cdot,\cdot\rangle$, i.e., twice its Wirtinger derivative with respect to $\Fm^\ast$, is
\begin{equation}\label{eq:Egrad_F}
    \nabla_\Fm\mathcal{J} = -2\Hm\,\Um\,\Omegam\,\Em\,\diag(\sqrt{\pv}),
\end{equation}
and the Riemannian gradient is its projection $\mathrm{grad}\,\mathcal{J}(\Fm) = \mathrm{P}_\Fm(\nabla_\Fm\mathcal{J})$.
At iteration $t$ of the inner loop, with $\mathbf{g}^{(t)}\triangleq\mathrm{grad}\,\mathcal{J}(\Fm^{(t)})$, we apply the \ac{QR} retraction along a Polak--Ribi\`ere \ac{CG} direction,
\begin{equation}\label{eq:F_iter}
\Fm^{(t+1)}\!=\!R_{\Fm^{(t)}}(\alpha^{(t)} \mathbf{d}^{(t)}),\ \mathbf{d}^{(t)}\!=\!-\mathbf{g}^{(t)}\!+\!\beta^{(t-1)} \mathrm{P}_{\Fm^{(t)}}(\mathbf{d}^{(t-1)}),
\end{equation}
where $\mathbf{d}^{(0)}=-\mathbf{g}^{(0)}$, the projection~\eqref{eq:Stiefel_proj} onto the tangent space at the new point serves as vector transport, and $\beta^{(t-1)}=\max\{0,\langle\mathbf{g}^{(t)},\mathbf{g}^{(t)}-\mathrm{P}_{\Fm^{(t)}}(\mathbf{g}^{(t-1)})\rangle/\|\mathbf{g}^{(t-1)}\|_F^2\}$ is the Polak--Ribi\`ere coefficient. We reset $\mathbf{d}^{(t)}=-\mathbf{g}^{(t)}$ whenever $\langle\mathbf{g}^{(t)},\mathbf{d}^{(t)}\rangle>-c_1\|\mathbf{g}^{(t)}\|_F^2$ or $\|\mathbf{d}^{(t)}\|_F>c_2\|\mathbf{g}^{(t)}\|_F$, with constants $c_1\in(0,1)$ and $c_2\ge 1$. The step size is $\alpha^{(t)}=\rho^m\bar{\alpha}/\|\mathbf{d}^{(t)}\|_F$, where $\bar\alpha>0$, $\rho\in(0,1)$, and $m$ is the smallest nonnegative integer that satisfies the Armijo condition $\mathcal{J}(R_{\Fm^{(t)}}(\alpha^{(t)}\mathbf{d}^{(t)}))\le\mathcal{J}(\Fm^{(t)})+c_{\mathrm{A}}\alpha^{(t)}\langle\mathbf{g}^{(t)},\mathbf{d}^{(t)}\rangle$ for a constant $c_{\mathrm{A}}\in(0,1)$~\cite[Def.~4.2.2]{Absil2008Optimization}. Unless $\mathbf{g}^{(0)}=\Zerom$, the inner loop performs at least one iteration, and it stops when $\|\mathbf{g}^{(t)}\|_F<\tau$ or after $T_{\mathrm{in}}$ iterations.

The same steps of tangent-space projection, retraction, and line search were used in~\cite{Yu2016AltMin} on the complex circle manifold for unit-modulus hybrid precoding, and we adapt them here to the Stiefel manifold.

\begin{proposition}[Convergence of the $\Fm$-update]\label{prop:F_update_convergence}
For fixed $(\pv,\uv,\omegav)$, the cost $\mathcal{J}(\cdot;\pv,\Um,\Omegam)$ in~\eqref{eq:F_subprob_stem} is smooth on the compact manifold $\St(N,K)$ and decreases at every inner iteration. The safeguards above make the normalized directions $\mathbf{d}^{(t)}/\|\mathbf{d}^{(t)}\|_F$ gradient-related~\cite[Def.~4.2.1]{Absil2008Optimization}, so every accumulation point of $\{\Fm^{(t)}\}$ is a critical point of~\eqref{eq:F_subprob_stem} and $\|\mathbf{g}^{(t)}\|_F\to 0$~\cite[Thm.~4.3.1, Cor.~4.3.2]{Absil2008Optimization}. Hence, for any $\tau>0$, the test $\|\mathbf{g}^{(t)}\|_F<\tau$ is met after finitely many iterations.
\end{proposition}

\subsection{Overall Algorithm, Convergence, and Complexity}
\label{sec:stem_overall}

Algorithm~\ref{alg:stem} summarizes the complete procedure for~\eqref{eq:Pstem}. \textbf{Initialization:} We start from equal powers $\pv^{(0)}=(P_T/K)\onev_K$ and from the regularized zero-forcing beamformer $\Wm_{\mathrm{rzf}}\triangleq\Hm(\overline{\Hm}+(K\sigma^2/P_T)\Id_K)^{-1}$ mapped to $\St(N,K)$. If $N\ge 2K-1$, $\Fm^{(0)}$ is obtained from $\Wm_{\mathrm{rzf}}$ by the construction in Step~3 of the proof of Proposition~\ref{prop:phase}, with the first $\min(N-K,K)$ rows of $\widetilde{\Vm}$ equal to $\sqrt{\lambda_i}\,\mathbf{q}_i^\HH$ for the eigenpairs $(\lambda_i,\mathbf{q}_i)$ of $\Id_K-\widetilde{\Ym}^\HH\widetilde{\Ym}$ in decreasing order and the other rows zero. If $N\le 2K-2$, $\Fm^{(0)}$ is the orthonormal polar factor of $\Wm_{\mathrm{rzf}}$, i.e., its nearest point on $\St(N,K)$, which equals $\Hm\overline{\Hm}^{-1/2}$.

\begin{algorithm}[t]
\caption{Proposed algorithm for the \ac{SC}-\ac{MiLAC}-aided beamforming problem~\eqref{eq:Pstem}.}
\label{alg:stem}
\footnotesize
\begin{algorithmic}[1]
\State \textbf{Input:} channel $\Hm$, power budget $P_T$, noise $\sigma^2$, inner tolerance $\tau$, outer tolerance $\epsilon$, inner iteration limit $T_{\mathrm{in}}$.
\State Initialize $\Fm^{(0)}$ from $\Wm_{\mathrm{rzf}}$ as described in Section~\ref{sec:stem_overall}, and $\pv^{(0)}=(P_T/K)\onev_K$.
\Repeat
\State Update $u_k$ and $\omega_k$ via~\eqref{eq:u_update_stem} for $k=1,\ldots,K$.
\State Update $\pv$ via~\eqref{eq:p_update_stem} with bisection on $\mu_p$.
\State Solve the $\Fm$-subproblem~\eqref{eq:F_subprob_stem} by iterating the Polak--Ribière \ac{CG} update~\eqref{eq:F_iter} with Armijo line search, at least once and until $\|\mathrm{grad}\,\mathcal{J}\|_F<\tau$ or $T_{\mathrm{in}}$ iterations.
\Until{$|\mathcal{J}^{(n)}-\mathcal{J}^{(n-1)}|\le\epsilon|\mathcal{J}^{(n-1)}|$, where $\mathcal{J}^{(n)}$ is the objective value of~\eqref{eq:wmmse_stem} after iteration $n$.}
\State Synthesize $\bv^\star\gets\Psi(\Fm^\star)$ via~\cite[Alg.~1]{Nerini2026Stem} and~\eqref{eq:Bmat_from_phys}.
\State \textbf{Return} $(\Fm^\star,\pv^\star,\bv^\star)$.
\end{algorithmic}
\end{algorithm}

\begin{proposition}[Convergence of Algorithm~\ref{alg:stem}]\label{prop:alg1_conv}
For any $\tau$ and $T_{\mathrm{in}}$, the iterates of Algorithm~\ref{alg:stem} satisfy:
\begin{enumerate}
\item[(a)] The sum-rate $\mathcal{R}^{(n)}\triangleq\mathcal{R}(\Fm^{(n)}\diag(\sqrt{\pv^{(n)}}))$ is non-decreasing in $n$ and converges, and the iterates $(\Fm^{(n)},\pv^{(n)})$ have at least one accumulation point.
\item[(b)] Every accumulation point $(\overline{\Fm},\overline{\pv})$ of the iterates is a stationary point of~\eqref{eq:Pstem}: $\overline{\Fm}$ is a critical point of $\mathcal{R}(\cdot\,\diag(\sqrt{\overline{\pv}}))$ on $\St(N,K)$, and $\sqrt{\overline{\pv}}$ satisfies the \ac{KKT} conditions of the maximization of $\mathcal{R}(\overline{\Fm}\diag(\xv))$ over the amplitudes $\xv\ge\zerov$ with $\|\xv\|^2\le P_T$.
\end{enumerate}
\end{proposition}

\begin{proof}
Part~(a) follows since each block update does not increase the objective of~\eqref{eq:wmmse_stem}, which equals $K-\mathcal{R}^{(n)}\ln 2$ by~\eqref{eq:wmmse_identity_stem} once $(u_k,\omega_k)$ are updated at $(\Fm^{(n)},\pv^{(n)})$, and since $\mathcal{R}^{(n)}$ is bounded by the fully digital optimum and the iterates lie in a compact set. Part~(b) is proved in the Appendix.
\end{proof}

Proposition~\ref{prop:alg1_conv} concerns accumulation points and does not establish that the whole sequence of iterates converges to a single point. Since the sum-rate depends on each amplitude $x_k=\sqrt{p_k}$ only through $x_k^2$, a stream with $\overline{p}_k=0$ satisfies its \ac{KKT} condition in~(b) for any $\overline{\Fm}$, so for such streams the condition is weaker than one stated in terms of $\pv$.

\textbf{Complexity:} Each line-search trial costs $\mathcal{O}(NK^2)$ for the \ac{QR} retraction and the evaluation of $\mathcal{J}$, and so do the gradient~\eqref{eq:Egrad_F}, the projection~\eqref{eq:Stiefel_proj}, the $(u_k,\omega_k)$ and $\pv$ updates, dominated by $\Hm^\HH\Fm$, and the initialization. The synthesis $\Psi$ costs $\mathcal{O}(NK^2)$ for $\Bm^\star$, since the $N-K+1$ linear systems of size $K$ in~\cite[Alg.~1]{Nerini2026Stem} share a common $K\times(K-1)$ block whose factorization is reused, plus $\mathcal{O}(NK)$ for the component values via~\eqref{eq:Bmat_from_phys}.

\section{Stem-Connected MiLAC-Aided MU-MISO Beamforming: Hardware-Compliant Case}
\label{sec:synthesis}

We now address the practical case in which each tunable susceptance takes finitely many values in a bounded interval, modeled by a finite grid (Section~\ref{sec:hw_model}). We propose two algorithms that trade complexity for sum-rate: a closed-form projection baseline (Section~\ref{sec:synth_baseline}) and a per-element alternating refinement (Section~\ref{sec:synth_ao}). The components remain lossless and the antennas uncoupled.

Both algorithms apply to both architectures. The architecture determines the $N_C\in\{N_C^{\mathrm{stem}},N_C^{\mathrm{full}}\}$ tunable components, and thus the map $\Phi$, as well as the ideal-case solver that produces $\Fm^\star$ and the synthesis $\Psi:\Fm^\star\mapsto\bv^\star\in\mathbb{R}^{N_C}$ of Section~\ref{sec:milac_arch}.

\subsection{Bounded-and-Discrete Susceptance Grid}
\label{sec:hw_model}

We model the susceptance values by a uniform grid.

\begin{definition}[Hardware susceptance grid]
\label{def:grid}
Fix a dynamic range $B>0$ and a number of levels $L\in\mathbb{N}$, $L\ge 2$. The \emph{hardware susceptance grid} is
\begin{equation}\label{eq:grid}
    \mathcal{B}_{B,L} \triangleq \left\{-B + \tfrac{2B(\ell-1)}{L-1} : \ell = 1,2,\ldots,L\right\}\!\subset\![-B,B].
\end{equation}
\end{definition}

We write $q\triangleq\log_2 L$ for the resolution in bits. The range is symmetric because the synthesis $\Psi$ returns susceptances of both signs. The grid spacing is $\Delta_{B,L}\triangleq 2B/(L-1)$.

\begin{remark}[Practical tunable components]
\label{rem:circuit}
The uniform symmetric grid~\eqref{eq:grid} is a baseline model. A common circuit combines a varactor, whose capacitance $C$ is set by a bias voltage, with fixed inductors. Neglecting its resistance, the circuit of~\cite{ZhouLossy}, which adopts the model of~\cite{Peng2026Lossy}, connects an inductance $L_1$ in parallel with an inductance $L_2$ in series with the varactor, and its susceptance at angular frequency $\omega_{\mathrm{c}}$ is
\begin{equation}\label{eq:tac}
b(C)=-\frac{1}{\omega_{\mathrm{c}} L_1}+\frac{\omega_{\mathrm{c}} C}{1-\omega_{\mathrm{c}}^2L_2C},
\end{equation}
which increases with $C$ and changes sign. With the values of~\cite{ZhouLossy}, $L_1=6$~nH, $L_2=0.7$~nH, $C\in[0.35,3.2]$~pF, and a frequency of $2.4$~GHz,~\eqref{eq:tac} ranges from $-5.46$~mS to $87.3$~mS. This range is asymmetric, but it contains the symmetric range $[-B,B]$ for $B\le 5.46$~mS. A $q$-bit control of the bias voltage, or a bank of $q$ binary-weighted switched capacitors, provides $2^q$ values of $C$. For $B=3.5$~mS, equally spaced capacitances between $0.46$~pF and $0.84$~pF map through~\eqref{eq:tac} to susceptances whose spacing deviates from uniform by less than $6\%$ for $q=3$ and by less than $7\%$ for $q\le 6$. The algorithms of Sections~\ref{sec:synth_baseline} and~\ref{sec:synth_ao} apply to any finite set of susceptance values, and Section~\ref{sec:sim_robust} evaluates an asymmetric range and the non-uniform spacing of~\eqref{eq:tac}. Since the beamformer depends on the susceptances only through $\Bm/Y_0$ (Section~\ref{sec:model}), the results on the grid hold for any reference impedance if $B$ is scaled in proportion to $Y_0$, e.g., $B/Y_0=0.175$ for $B=3.5$~mS and $Y_0=0.02$~S. Since~\eqref{eq:tac} depends on $\omega_{\mathrm{c}}$, each susceptance also varies over a wide band. Like the wideband study~\cite{Peng2026OFDM}, which uses one \ac{MiLAC} beamforming matrix for all subcarriers, we do not model this variation.
\end{remark}

The projection onto the grid is
\begin{equation}\label{eq:proj_grid}
    \mathrm{proj}_{\mathcal{B}_{B,L}}(b)\triangleq\arg\min_{g\in\mathcal{B}_{B,L}} |g-b|,\qquad b\in\mathbb{R},
\end{equation}
which clips $b$ to $[-B,B]$ and rounds it to the nearest grid point. We write $\mathrm{proj}_{\mathcal{B}_{B,L}}(\bv)$ for the entrywise application of this map to a vector. The feasible susceptance vectors form the set $\mathcal{B}_{B,L}^{N_C}$, and the hardware-constrained problem is
\begin{equation}\label{eq:P_hw}
\mathcal{P}_{B,L}:\;\max_{\bv\in\mathcal{B}_{B,L}^{N_C},\,\pv\in\mathbb{R}_+^K,\,\onev_K^\TT\pv\le P_T}\;\mathcal{R}\bigl(\Phi(\bv)\diag(\sqrt{\pv})\bigr).
\end{equation}

\subsection{Post-Hoc Projection Baseline}
\label{sec:synth_baseline}

\begin{algorithm}[t]
\caption{\ac{PHP} synthesis for problem \eqref{eq:P_hw}.}
\label{alg:proj}
\footnotesize
\begin{algorithmic}[1]
\State \textbf{Input:} channel $\Hm$, power budget $P_T$, grid $\mathcal{B}_{B,L}$.
\State Obtain $(\Fm^\star,\pv^\star)$ from the architecture's ideal-case solver (Algorithm~\ref{alg:stem} for \ac{SC}-\ac{MiLAC}, \cite[Alg.~2]{Wu2026MiLAC} for \ac{FC}-\ac{MiLAC}) and set $\bv^\star\gets\Psi(\Fm^\star)$.
\State Project each entry: $b_i^{\mathrm{q}}\gets\mathrm{proj}_{\mathcal{B}_{B,L}}(b_i^\star)$ for $i=1,\ldots,N_C$.
\State Form $\Fm^{\mathrm{q}}\gets\Phi(\bv^{\mathrm{q}})$ and keep $\pv^\star$ unchanged.
\State \textbf{Return} feasible $(\bv^{\mathrm{q}},\pv^\star)$.
\end{algorithmic}
\end{algorithm}

Algorithm~\ref{alg:proj} is the \ac{PHP} baseline: it projects each susceptance of the ideal solution onto the grid, maps the result $\bv^{\mathrm{q}}$ through~\eqref{eq:cayley} to the realized beamformer $\Fm^{\mathrm{q}}=\Phi(\bv^{\mathrm{q}})$, and keeps the ideal power allocation $\pv^\star$. Besides the ideal-case solver and the synthesis~$\Psi$, it costs one $(N{+}K)\times(N{+}K)$ inversion. Since the projection ignores its effect on the sum-rate, the \ac{AR} of Section~\ref{sec:synth_ao} refines its output.

\subsection{Per-Element Alternating Refinement}
\label{sec:synth_ao}

We now optimize directly over the grid in~\eqref{eq:P_hw}. Since the finite set $\mathcal{B}_{B,L}^{N_C}$ rules out the continuous $\Fm$-update of Section~\ref{sec:wmmse}, we retain the \ac{WMMSE} framework with $\Fm=\Phi(\bv)$ and replace only the $\Fm$-update by a per-element grid scan of $\bv$. The updates~\eqref{eq:u_update_stem}--\eqref{eq:p_update_stem} apply without modification.

\subsubsection{Update of $\bv$ over the Susceptance Grid}
\label{sec:ao_b_scan}

For fixed $\pv$, the $\bv$-update subproblem is
\begin{equation}\label{eq:b_subprob}
    \max_{\bv\in\mathcal{B}_{B,L}^{N_C}}\ \mathcal{R}\bigl(\Phi(\bv)\diag(\sqrt{\pv})\bigr).
\end{equation}
Problem~\eqref{eq:b_subprob} is combinatorial, since $\bv$ takes one of $L^{N_C}$ values. We address~\eqref{eq:b_subprob} by a per-element grid scan: at each step, a single susceptance $b_i$ is updated over its $L$-point grid while the remaining $N_C-1$ entries are held fixed. Let $(\bv_{-i},b)$ denote $\bv$ with its $i$th entry replaced by $b\in\mathbb{R}$. The sum-rate as a function of the $i$th susceptance,
\begin{equation}\label{eq:section}
    \mathcal{R}_i(b)\triangleq\mathcal{R}\bigl(\Phi(\bv_{-i},b)\diag(\sqrt{\pv})\bigr),\quad b\in\mathcal{B}_{B,L},
\end{equation}
takes $L$ values, so its maximum is found by enumeration,
\begin{equation}\label{eq:b_update}
    b_i\leftarrow\arg\max_{b\in\mathcal{B}_{B,L}}\mathcal{R}_i(b),
\end{equation}
where the current value is kept unless another value is strictly better. The susceptances are updated cyclically, $i=1,\ldots,N_C$, and each update is applied before the next scan. We score the candidates by the sum-rate~\eqref{eq:rate} rather than by the \ac{WMMSE} objective $\mathcal{J}$ in~\eqref{eq:Jsurr_stem}, since $\mathcal{J}$ equals $K-\mathcal{R}\ln 2$ only at the $(u_k,\omega_k)$ computed for the current beamformer, and every update of a susceptance changes the beamformer, whereas the sum-rate depends only on $(\bv,\pv)$.

\begin{algorithm}[t]
\caption{Per-element \ac{AR} for problem~\eqref{eq:P_hw}.}
\label{alg:ao}
\footnotesize
\begin{algorithmic}[1]
\State \textbf{Input:} $\Hm$, $P_T$, grid $\mathcal{B}_{B,L}$, max outer iterations $T_{\mathrm{out}}$, tolerance $\epsilon$.
\State Initialize $(\bv^{(0)},\pv^{(0)})\gets(\bv^{\mathrm{q}},\pv^\star)$ from Algorithm~\ref{alg:proj}.
\For{$t=0,1,\ldots,T_{\mathrm{out}}-1$}
\For{$i=1,\ldots,N_C$}
\State Set $b_i\leftarrow\arg\max_{b\in\mathcal{B}_{B,L}}\mathcal{R}_i(b)$ at $\pv^{(t)}$ as in~\eqref{eq:b_update}, evaluating each candidate via~\eqref{eq:smw}--\eqref{eq:Phi_block_update}.
\EndFor
\State Update $u_k$, $\omega_k$ via~\eqref{eq:u_update_stem} at $\Fm^{(t+1)}=\Phi(\bv^{(t+1)})$ and $\pv^{(t)}$.
\State Update $\pv^{(t+1)}$ via~\eqref{eq:p_update_stem} with bisection on $\mu_p$.
\State \textbf{Break} if $\mathcal{R}^{(t+1)}-\mathcal{R}^{(t)}\le\epsilon\,\mathcal{R}^{(t)}$, with $\mathcal{R}^{(t)}\triangleq\mathcal{R}(\Phi(\bv^{(t)})\diag(\sqrt{\pv^{(t)}}))$.
\EndFor
\State \textbf{Return} $(\bv^{(t+1)},\pv^{(t+1)})$.
\end{algorithmic}
\end{algorithm}

\subsubsection{Reduced-Complexity Implementation}
\label{sec:ao_smw}

Evaluating~\eqref{eq:section} directly requires an $(N+K)\times(N+K)$ matrix inversion per candidate $b$, which we avoid as follows. By~\eqref{eq:Bmat_from_phys}, replacing $b_i$ with $b$ perturbs $\Bm$ by a symmetric rank-one term,
\begin{equation}\label{eq:rank1pert}
\Bm'=\Bm+\delta_i\vv^{(i)}\vv^{(i)\TT},\qquad \delta_i\triangleq b-b_i,
\end{equation}
where $\vv^{(i)}=\ev_{m_i}-\ev_{m_i'}$ if component $i$ connects ports $m_i$ and $m_i'$, and $\vv^{(i)}=\ev_{m_i}$ if it connects port $m_i$ to ground, with $\ev_m$ the $m$th standard basis vector of $\mathbb{R}^{N+K}$. With the symmetric resolvent $\Mm\triangleq(Y_0\Id+j\Bm)^{-1}$, the Sherman--Morrison identity gives the updated resolvent
\begin{equation}\label{eq:smw}
\Mm'\triangleq(Y_0\Id+j\Bm')^{-1}=\Mm-\gamma_i(\Mm\vv^{(i)})(\Mm\vv^{(i)})^{\TT},
\end{equation}
where $\gamma_i\triangleq j\delta_i/(1+j\delta_i\vv^{(i)\TT}\Mm\vv^{(i)})$. The denominator of $\gamma_i$ equals $\det(Y_0\Id+j\Bm')/\det(Y_0\Id+j\Bm)$ by the matrix determinant lemma, and it never vanishes, because every eigenvalue of $Y_0\Id+j\Bm'$ has real part $Y_0>0$. Since $\Thetam=2Y_0\Mm-\Id$ by~\eqref{eq:cayley}, the beamformer is $\Phi(\bv)=2Y_0[\Mm]_{K+1:K+N,\,1:K}$, and~\eqref{eq:smw} gives each candidate beamformer as
\begin{equation}\label{eq:Phi_block_update}
\Phi(\bv_{-i},b)=\Phi(\bv)-2Y_0\gamma_i[\Mm\vv^{(i)}]_{K+1:K+N}[\Mm\vv^{(i)}]_{1:K}^{\TT},
\end{equation}
i.e., by a rank-one correction to the cached $N\times K$ block $\Phi(\bv)$ rather than by a fresh inversion. Premultiplying~\eqref{eq:Phi_block_update} by $\Hm^\HH$ updates the effective channel $\Gm\triangleq\Hm^\HH\Phi(\bv)\in\mathbb{C}^{K\times K}$ by a rank-one term, from which $\mathcal{R}_i(b)$ follows at $\mathcal{O}(K^2)$ cost.

\subsection{Overall Algorithm, Convergence, and Complexity}
\label{sec:ao_overall}

Algorithm~\ref{alg:ao} alternates the sweep~\eqref{eq:b_update} at fixed $\pv$ with the updates~\eqref{eq:u_update_stem}--\eqref{eq:p_update_stem} at the resulting beamformer, starting from the \ac{PHP} point of Algorithm~\ref{alg:proj}. The sweep precedes the power update, so that $\pv$ is adapted to the refined beamformer rather than to the projected \ac{PHP} one. The sum-rate never decreases: a sweep keeps each current value unless another one is strictly better, and the power update cannot decrease it by the argument in the proof of Proposition~\ref{prop:alg1_conv}(a). Being non-decreasing and bounded, the sum-rate converges, so for every $\epsilon>0$ the stopping test of Algorithm~\ref{alg:ao} is met after finitely many iterations if $T_{\mathrm{out}}$ is large enough. The output need not be a global optimum of the discrete problem~\eqref{eq:P_hw}, but the last sweep and power update together raised the sum-rate by at most a fraction~$\epsilon$ of its previous value.

\textbf{Complexity:} For each component, forming $\Mm\vv^{(i)}$ costs $\mathcal{O}(N+K)$, since $\vv^{(i)}$ has at most two nonzero entries, and $\Hm^\HH[\Mm\vv^{(i)}]_{K+1:K+N}$ costs $\mathcal{O}(NK)$. Each of the $L$ candidates then costs $\mathcal{O}(K^2)$ for the rank-one update of $\Gm$ and the sum-rate, and committing the best candidate refreshes $\Mm$, $\Phi(\bv)$, and $\Gm$ by rank-one updates at $\mathcal{O}((N+K)^2)$. A sweep therefore costs $\mathcal{O}\bigl(N_C[(N+K)^2+LK^2]\bigr)$, instead of $\mathcal{O}(N_CL(N+K)^3)$ with one inversion per candidate. We recompute $\Mm$ by one full inversion per sweep to limit rounding errors, and the power update costs $\mathcal{O}(K^2)$ given $\Gm$. The total cost of \ac{AR} with at most $T_{\mathrm{out}}$ sweeps is thus $\mathcal{O}\bigl(T_{\mathrm{out}}N_C[(N+K)^2+LK^2]\bigr)$ in addition to its \ac{PHP} initialization.

\section{Numerical Results}
\label{sec:sim}

We report Monte Carlo simulations over Rician channels,
\begin{equation}
\hv_k=\sqrt{\frac{\kappa}{\kappa+1}}\,\mathbf{a}(\theta_k) + \sqrt{\frac{1}{\kappa+1}}\,\mathbf{g}_k,
\end{equation}
where $\kappa=10^{0.6}\approx 3.98$ is the Rician factor ($6$~dB), $\mathbf{a}(\theta)\in\mathbb{C}^N$ is the steering vector of a uniform linear array with half-wavelength spacing, $[\mathbf{a}(\theta)]_n=e^{j\pi(n-1)\sin\theta}$ for $n=1,\ldots,N$, the angles $\theta_k\sim\mathrm{Unif}(-\pi/2,\pi/2)$ are i.i.d.\ across users, and $\mathbf{g}_k\sim\mathcal{CN}(\Zerom,\Id_N)$ is the \ac{NLOS} component. The \ac{SNR} is $P_T/\sigma^2$. Unless otherwise specified, $N=64$, $K=4$, $\mathrm{SNR}=10$~dB, $q=3$, and $B=3.5$~mS, with $Y_0=(50\,\Omega)^{-1}=0.02$~S. Every point averages the same $200$ channel realizations.\footnote{We report mean sum-rates, whose $95\%$ confidence intervals have half-widths below $0.6$~bits/s/Hz. Since all schemes use the same channels, their differences vary less from channel to channel and have narrower intervals.} All schemes, including the \ac{PS} baselines, use $N_{\mathrm{RF}}=K$ \ac{RF} chains and the power budget $P_T$. Algorithm~\ref{alg:stem} uses $\tau=10^{-4}$, $T_{\mathrm{in}}=500$, $\epsilon=10^{-6}$, $\bar\alpha=1$, $\rho=0.5$, $c_{\mathrm{A}}=c_1=10^{-4}$, $c_2=10^{4}$, and at most $5000$ outer iterations. Algorithm~\ref{alg:ao} uses $T_{\mathrm{out}}=200$ and $\epsilon=10^{-6}$, and the solver of~\cite[Alg.~2]{Wu2026MiLAC} uses inner and outer tolerances $10^{-7}$ with at most $2000$ iterations each. \ac{WMMSE} for Digital and FC-PS-Hybrid stops at a relative change of the sum-rate below $10^{-10}$ or after $2000$ iterations. Runtimes are for single-threaded code on an Apple M5 processor.

\begin{figure}[t]
\centering
\begin{tikzpicture}
\begin{axis}[
    width=1\linewidth, height=4.7cm,
    xlabel={\ac{SNR} (dB)}, ylabel={Sum-rate (bits/s/Hz)},
    xmin=-5, xmax=25, ymin=0, ymax=68,
    grid=both, grid style={black!12},
    legend style={font=\scriptsize, at={(0.02,0.98)}, anchor=north west, draw=black!30}, legend cell align={left}, legend columns=2,
    tick label style={font=\scriptsize}, label style={font=\small},
    every axis plot/.append style={line width=0.8pt, mark size=1.6pt}
]
\addplot[blue, mark=*] table[x=snr_db, y=Digital] {data/sumrate_vs_snr_ideal.dat};
\addlegendentry{Digital}
\addplot[purple, mark=o, mark options={fill=white}] table[x=snr_db, y=FCMiLAC] {data/sumrate_vs_snr_ideal.dat};
\addlegendentry{FC-MiLAC}
\addplot[red, mark=square*, densely dashed, mark options={solid}] table[x=snr_db, y=SCMiLAC] {data/sumrate_vs_snr_ideal.dat};
\addlegendentry{SC-MiLAC}
\addplot[teal!70!black, mark=star] table[x=snr_db, y=FCPSHybrid] {data/sumrate_vs_snr_ideal.dat};
\addlegendentry{FC-PS-Hybrid}
\addplot[orange!85!black, mark=triangle*] table[x=snr_db, y=FCPS] {data/sumrate_vs_snr_ideal.dat};
\addlegendentry{FC-PS}
\end{axis}
\end{tikzpicture}
\vspace{-3mm}
\caption{Sum-rate versus \ac{SNR} when $N=64$ and $K=4$.}
\label{fig:snr}
\end{figure}

\begin{figure}[t]
\centering
\begin{tikzpicture}
\begin{axis}[
    width=1\linewidth, height=4.7cm,
    xlabel={\ac{SNR} (dB)}, ylabel={Sum-rate (bits/s/Hz)},
    xmin=-5, xmax=25, ymin=0, ymax=44,
    grid=both, grid style={black!12},
    legend style={font=\scriptsize, at={(0.02,0.98)}, anchor=north west, draw=black!30}, legend cell align={left}, legend columns=2,
    tick label style={font=\scriptsize}, label style={font=\small},
    every axis plot/.append style={line width=0.8pt, mark size=1.6pt}
]
\addplot[blue, mark=*] table[x=snr_db, y=Digital] {data/sumrate_vs_snr_edge_ideal.dat};
\addlegendentry{Digital}
\addplot[purple, mark=o, mark options={fill=white}] table[x=snr_db, y=FCMiLAC] {data/sumrate_vs_snr_edge_ideal.dat};
\addlegendentry{FC-MiLAC}
\addplot[red, mark=square*, densely dashed, mark options={solid}] table[x=snr_db, y=SCMiLAC] {data/sumrate_vs_snr_edge_ideal.dat};
\addlegendentry{SC-MiLAC}
\addplot[black, mark=x, only marks, mark size=2.4pt] table[x=snr_db, y=SCsynth] {data/sumrate_vs_snr_edge_ideal.dat};
\addlegendentry{SC-MiLAC, synthesis}
\addplot[teal!70!black, mark=star] table[x=snr_db, y=FCPSHybrid] {data/sumrate_vs_snr_edge_ideal.dat};
\addlegendentry{FC-PS-Hybrid}
\addplot[orange!85!black, mark=triangle*] table[x=snr_db, y=FCPS] {data/sumrate_vs_snr_edge_ideal.dat};
\addlegendentry{FC-PS}
\end{axis}
\end{tikzpicture}
\vspace{-3mm}
\caption{Sum-rate versus \ac{SNR} when $N=6$ and $K=4$.}
\label{fig:snr_edge}
\end{figure}

\subsection{Ideal Case}
\label{sec:sim_ideal}

We compare the following schemes:
\begin{itemize}
    \item \emph{Digital}, \emph{FC-MiLAC}, and \emph{SC-MiLAC:} Problems~\eqref{eq:Pdig},~\eqref{eq:Pfull}, and~\eqref{eq:Pstem} solved by \ac{WMMSE}~\cite{Shi2011WMMSE}, by~\cite[Alg.~2]{Wu2026MiLAC}, and by the proposed Algorithm~\ref{alg:stem}, respectively.
    \item \emph{FC-PS-Hybrid:} Hybrid beamforming $\Wm=\Fm_{\mathrm{RF}}\Fm_{\mathrm{BB}}$ with $NK$ \acp{PS}. Each entry of $\Fm_{\mathrm{RF}}\in\mathbb{C}^{N\times K}$ has modulus $1/\sqrt{N}$ and the phase of the corresponding entry of the Digital solution. The digital beamformer $\Fm_{\mathrm{BB}}\in\mathbb{C}^{K\times K}$ is computed by \ac{WMMSE}~\cite{Shi2011WMMSE} under $\tr(\Wm\Wm^\HH)\le P_T$, which the substitution $\Fm_{\mathrm{BB}}=(\Fm_{\mathrm{RF}}^\HH\Fm_{\mathrm{RF}})^{-1/2}\Xm$ turns into $\tr(\Xm\Xm^\HH)\le P_T$.
    \item \emph{FC-PS:} The \ac{PS}-based counterpart of \ac{MiLAC}-aided beamforming, $\Wm=\Fm\diag(\sqrt{\pv})$ with $\onev_K^\TT\pv\le P_T$ as in~\eqref{eq:WfromF}, where a network of $NK$ \acp{PS} realizes $\Fm\in\mathbb{C}^{N\times K}$ with entries of modulus $1/\sqrt{NK}$. Then $\|\Fm\|_2\le\|\Fm\|_F=1$ for every setting of the phases, so the \ac{PS} network does not amplify its input~\cite{Wu2026MiLAC}, as for the \ac{MiLAC}. This is the largest common modulus with this property, since equal phases give $\|\Fm\|_2=\|\Fm\|_F$. We compute $(\Fm,\pv)$ by Algorithm~\ref{alg:stem} with the $\Fm$-update on the manifold of such matrices~\cite{Yu2016AltMin} instead of $\St(N,K)$, starting from the phases of the Digital solution.
\end{itemize}

FC-PS is thus fed by the same $K$ amplifiers as the \ac{MiLAC} designs under the same budget $\onev_K^\TT\pv\le P_T$. FC-PS-Hybrid follows the common model of hybrid beamforming, in which $P_T$ bounds the radiated power $\tr(\Wm\Wm^\HH)$, as for Digital~\cite{Sohrabi2016Hybrid,Yu2016AltMin}. For the \ac{SC}-\ac{MiLAC} designs with ideal susceptances, $\Fm^\HH\Fm=\Id_K$, so the two budgets coincide by~\eqref{eq:power_split}.

Figure~\ref{fig:snr} sweeps the \ac{SNR} for $N=64$, so that $N\ge 2K-1$. \ac{SC}-\ac{MiLAC} and \ac{FC}-\ac{MiLAC} reach the same average sum-rate, as Proposition~\ref{prop:phase} predicts, while the \ac{SC}-\ac{MiLAC} uses $516$ instead of $2346$ tunable components. Their mean difference is at most $0.06$~bits/s/Hz at every \ac{SNR}, and their sum-rates coincide within $0.02$~bits/s/Hz in $1389$ of the $1400$ cases ($200$ channels at seven \acp{SNR}). The other $11$ cases are different local solutions, in which the solver of~\cite[Alg.~2]{Wu2026MiLAC} allocates zero power to one user. Both \ac{MiLAC} curves stay within $1.3$~bits/s/Hz of Digital. FC-PS stays $5.7$ to $7.5$~bits/s/Hz below \ac{SC}-\ac{MiLAC} due to the limited flexibility of the \ac{PS} network compared with the \ac{MiLAC}. FC-PS-Hybrid performs close to the \ac{MiLAC} designs. \ac{SC}-\ac{MiLAC} exceeds it by $0.36$ and $0.10$~bits/s/Hz at $-5$ and $0$~dB, whereas FC-PS-Hybrid is ahead by $0.26$ to $0.63$~bits/s/Hz from $5$~dB on, consistent with the comparison of \ac{FC}-\ac{MiLAC}-aided and hybrid beamforming in~\cite{Wu2026MiLAC}.

Figure~\ref{fig:snr_edge} considers $(N,K)=(6,4)$, where $N\le 2K-2$ (Remark~\ref{rem:binding}). The mean difference between \ac{SC}-\ac{MiLAC} and \ac{FC}-\ac{MiLAC} is below $0.005$~bits/s/Hz up to $10$~dB and grows to $0.16$ and $0.55$~bits/s/Hz at $20$ and $25$~dB. To test whether this gap is caused by the \ac{SC} topology, we solved $\min_{\bv}\|\Phi(\bv)-\Fm\|_F$ over the $52$ susceptances of the \ac{SC}-\ac{MiLAC} for each of the $1400$ \ac{FC}-\ac{MiLAC} solutions $\Fm$ ($200$ channel realizations at each of the seven \acp{SNR}), by a trust-region least-squares method with up to seven starting points. The residual was below $10^{-8}$ in $1362$ cases, and the resulting \ac{SC}-\ac{MiLAC} attained the sum-rate of $\Fm$ within $0.03$~bits/s/Hz in all cases. This \ac{SC}-\ac{MiLAC}, denoted by ``SC-MiLAC, synthesis,'' coincides with \ac{FC}-\ac{MiLAC}, which indicates that the gap is caused by the Stiefel restriction and not by the \ac{SC} topology. Unlike the closed-form synthesis of beamformers with orthonormal columns, this numerical synthesis may not find an exact solution, since the least-squares problem is non-convex. The \ac{PS} baselines rank as in Fig.~\ref{fig:snr}, with a larger lead of FC-PS-Hybrid at high \ac{SNR}, as the user channels of a small array are more strongly correlated.

\begin{figure}[t]
\centering
\begin{tikzpicture}
\begin{axis}[
    width=1\linewidth, height=4.7cm,
    xlabel={Number of antennas $N$}, ylabel={Sum-rate (bits/s/Hz)},
    xmode=log, log basis x=2, xmin=4, xmax=256, ymin=0, ymax=56,
    xtick={4,5,6,7,8,16,32,64,128,256}, xticklabels={4,5,6,7,8,16,32,64,128,256},
    grid=both, grid style={black!12},
    legend style={font=\scriptsize, at={(0.02,0.98)}, anchor=north west, draw=black!30}, legend cell align={left}, legend columns=2,
    tick label style={font=\scriptsize}, label style={font=\small},
    every axis plot/.append style={line width=0.8pt, mark size=1.6pt}
]
\addplot[blue, mark=*] table[x=N, y=Digital] {data/sumrate_vs_N_ideal.dat};
\addlegendentry{Digital}
\addplot[purple, mark=o, mark options={fill=white}] table[x=N, y=FCMiLAC] {data/sumrate_vs_N_ideal.dat};
\addlegendentry{FC-MiLAC}
\addplot[red, mark=square*, densely dashed, mark options={solid}] table[x=N, y=SCMiLAC] {data/sumrate_vs_N_ideal.dat};
\addlegendentry{SC-MiLAC}
\addplot[teal!70!black, mark=star] table[x=N, y=FCPSHybrid] {data/sumrate_vs_N_ideal.dat};
\addlegendentry{FC-PS-Hybrid}
\addplot[orange!85!black, mark=triangle*] table[x=N, y=FCPS] {data/sumrate_vs_N_ideal.dat};
\addlegendentry{FC-PS}
\end{axis}
\end{tikzpicture}
\vspace{-3mm}
\caption{Sum-rate versus $N$ when $K=4$ and $\mathrm{SNR}=10$~dB.}
\label{fig:N}
\end{figure}

Figure~\ref{fig:N} sweeps $N$ at $\mathrm{SNR}=10$~dB. For every $N$, the mean difference between \ac{SC}-\ac{MiLAC} and \ac{FC}-\ac{MiLAC} is below $0.03$~bits/s/Hz, and for $N\ge 2K-1=7$ the two coincide within $0.02$~bits/s/Hz in at least $91\%$ of the realizations. Their gap to Digital is $8\%$ to $9\%$ of the Digital sum-rate for $N\le 16$ and falls to $1.5\%$ at $N=256$, as the user channels become nearly orthogonal~\cite{Wu2026MiLAC}. FC-PS-Hybrid exceeds \ac{SC}-\ac{MiLAC} by at most $1.05$~bits/s/Hz ($N=16$), attains the Digital sum-rate at $N=K=4$, where $\Fm_{\mathrm{RF}}$ is square, and matches \ac{SC}-\ac{MiLAC} at $N=256$. FC-PS stays $4.0$ to $8.0$~bits/s/Hz below \ac{SC}-\ac{MiLAC}.

\subsection{Hardware-Compliant Case}
\label{sec:sim_hw}

We compare the following schemes on the susceptance grid, together with the Digital benchmark:
\begin{itemize}
    \item \emph{Unc:} The sum-rate of Algorithm~\ref{alg:stem} with unconstrained susceptances, which serves as the ideal-case reference and is within $0.06$~bits/s/Hz of \ac{FC}-\ac{MiLAC} on average.
    \item \emph{SC-MiLAC-PHP} and \emph{SC-MiLAC-AR:} Algorithms~\ref{alg:proj} and~\ref{alg:ao} for \ac{SC}-\ac{MiLAC}, on $N_C^{\mathrm{stem}}=K(2N{+}1)$ tunable components, with \ac{AR} started from \ac{PHP}.
    \item \emph{FC-MiLAC-PHP} and \emph{FC-MiLAC-AR:} The \ac{FC} counterparts, on $N_C^{\mathrm{full}}=(N{+}K)(N{+}K{+}1)/2$ components.
    \item \emph{FC-PS-Hybrid} and \emph{FC-PS}, \emph{$q$-bit:} The \ac{PS} baselines of Section~\ref{sec:sim_ideal} with $q$-bit \acp{PS}, whose phases lie on the uniform grid $\{2\pi\ell/2^q\}_{\ell=0}^{2^q-1}$. For FC-PS-Hybrid, the phases of $\Fm_{\mathrm{RF}}$ are rounded to the grid and $\Fm_{\mathrm{BB}}$ is recomputed. For FC-PS, the rounded phases of the continuous solution are refined by cyclic sweeps over the $NK$ \acp{PS} that score each candidate phase by the sum-rate, alternated with the updates~\eqref{eq:u_update_stem}--\eqref{eq:p_update_stem}, as in Algorithm~\ref{alg:ao}.
\end{itemize}

\begin{figure}[t]
\centering
\begin{tikzpicture}
\begin{axis}[
    width=1\linewidth, height=4.7cm,
    xlabel={\ac{SNR} (dB)}, ylabel={Sum-rate (bits/s/Hz)},
    xmin=-5, xmax=25, ymin=0, ymax=72,
    grid=both, grid style={black!12},
    legend style={font=\scriptsize, at={(0.02,0.98)}, anchor=north west,
                  draw=black!30, fill=white, fill opacity=0.95,
                  text opacity=1},
    legend cell align={left}, legend columns=2,
    tick label style={font=\scriptsize}, label style={font=\small},
    every axis plot/.append style={line width=0.8pt, mark size=1.6pt}
]
\addplot[blue, mark=*] table[x=snr_db, y=Digital] {data/sumrate_vs_snr_hw.dat};
\addlegendentry{Digital}
\addplot[purple, no marks, dashed] table[x=snr_db, y=SCUnc] {data/sumrate_vs_snr_hw.dat};
\addlegendentry{Unc}
\addplot[red, mark=triangle*] table[x=snr_db, y=SCAR] {data/sumrate_vs_snr_hw.dat};
\addlegendentry{SC-MiLAC-AR}
\addplot[orange!85!black, mark=diamond*, densely dashed, mark options={solid}] table[x=snr_db, y=SCPHP] {data/sumrate_vs_snr_hw.dat};
\addlegendentry{SC-MiLAC-PHP}
\addplot[teal, mark=square*] table[x=snr_db, y=FCAR] {data/sumrate_vs_snr_hw.dat};
\addlegendentry{FC-MiLAC-AR}
\addplot[brown, mark=o, mark options={solid, fill=white}, densely dashed] table[x=snr_db, y=FCPHP] {data/sumrate_vs_snr_hw.dat};
\addlegendentry{FC-MiLAC-PHP}
\addplot[teal!70!black, mark=star, dotted, mark options={solid}] table[x=snr_db, y=PSH] {data/sumrate_vs_snr_hw.dat};
\addlegendentry{FC-PS-Hybrid, $3$-bit}
\addplot[gray, mark=triangle, dotted, mark options={solid}] table[x=snr_db, y=PS] {data/sumrate_vs_snr_hw.dat};
\addlegendentry{FC-PS, $3$-bit}
\end{axis}
\end{tikzpicture}
\vspace{-3mm}
\caption{Sum-rate versus \ac{SNR} when $N=64$, $K=4$, $q=3$, and $B=3.5$~mS.}
\label{fig:sr_snr_hw}
\end{figure}

Figure~\ref{fig:sr_snr_hw} compares the hardware-compliant designs with the unconstrained benchmark Unc for $q=3$ and $B=3.5$~mS. At $10$~dB, \ac{SC}-\ac{MiLAC}-\ac{AR} and \ac{FC}-\ac{MiLAC}-\ac{AR} attain $90.2\%$ and $94.7\%$ of Unc, respectively, and at $25$~dB $73.0\%$ and $81.0\%$. The gaps widen at high \ac{SNR}, since the \ac{AR} beamformers leave residual inter-user interference, which grows with $P_T$ and thus limits the \acp{SINR}. Unlike in the ideal case, in which \ac{SC}-\ac{MiLAC} and \ac{FC}-\ac{MiLAC} attain the same sum-rate, a gap appears between their designs on the grid: \ac{FC}-\ac{MiLAC}-\ac{AR} is ahead by $1.24$~bits/s/Hz at $10$~dB but uses $4.5$ times as many tunable components. The \ac{SC}-\ac{MiLAC} thus trades a moderate sum-rate loss for a large reduction in the number of tunable components. Both \ac{AR} designs exceed the $3$-bit FC-PS up to $15$~dB. The $3$-bit FC-PS-Hybrid, which adds a baseband digital beamformer, stays within $0.9$~bits/s/Hz of Digital and exceeds both \ac{AR} designs from $5$~dB on.

\textbf{Source of the hardware loss:} The loss of \ac{AR} relative to Unc can come from the grid, which may contain no configuration whose sum-rate is close to Unc, or from the local search, which stops where no change of a single susceptance increases the sum-rate. At $10$~dB, the default \ac{PHP} initialization clips $80\%$ of the synthesized susceptances of the \ac{SC}-\ac{MiLAC} and $56\%$ of those of the \ac{FC}-\ac{MiLAC}, which yields initial sum-rates of only $18\%$ and $27\%$ of Unc, respectively. To separate the two sources, we first run \ac{AR} on a finer $6$-bit grid, project the resulting solution onto the $3$-bit grid, and use this projection to initialize the same $3$-bit \ac{AR} algorithm. At $10$~dB, \ac{AR} from this start reaches $95.3\%$ (\ac{SC}) and $99.2\%$ (\ac{FC}) of Unc, instead of $90.2\%$ and $94.7\%$. Hence the $3$-bit grid itself costs at most $4.7\%$ (\ac{SC}) and $0.8\%$ (\ac{FC}) of Unc, and the starting point causes at least half of the loss of \ac{SC}-\ac{MiLAC}-\ac{AR} and most of that of \ac{FC}-\ac{MiLAC}-\ac{AR}. At $25$~dB, this start raises the two designs from $73.0\%$ and $81.0\%$ to $84.4\%$ and $91.0\%$ of Unc. From this start, \ac{FC}-\ac{MiLAC}-\ac{AR} still outperforms \ac{SC}-\ac{MiLAC}-\ac{AR}, so the advantage of the \ac{FC}-\ac{MiLAC} on the grid is not an artifact of the search. Both designs also exceed the $3$-bit FC-PS at every \ac{SNR}, whereas the $3$-bit FC-PS-Hybrid still leads both from $5$~dB on. We use this start only as a diagnostic, since computing it takes more than ten times as long as \ac{AR} (Section~\ref{sec:sim_cost}).

\begin{figure}[t]
\centering
\begin{tikzpicture}
\begin{axis}[
    width=1\linewidth, height=4.7cm,
    xmode=log, log basis x={10},
    xlabel={Dynamic range $B$ (mS)}, ylabel={Sum-rate (bits/s/Hz)},
    xmin=0.2, xmax=250, ymin=0, ymax=32,
    grid=both, grid style={black!12},
    legend style={font=\scriptsize, at={(0.5,1.03)}, anchor=south, draw=black!30},
    legend cell align={left}, legend columns=3,
    tick label style={font=\scriptsize}, label style={font=\small},
    every axis plot/.append style={line width=0.8pt, mark size=1.6pt}
]
\addplot[purple, no marks, dashed] table[x expr=\thisrow{B}*1000, y=SCUnc] {data/sumrate_vs_B.dat};
\addlegendentry{Unc}
\addplot[red, mark=triangle*] table[x expr=\thisrow{B}*1000, y=SCAR] {data/sumrate_vs_B.dat};
\addlegendentry{SC-MiLAC-AR}
\addplot[teal, mark=square*] table[x expr=\thisrow{B}*1000, y=FCAR] {data/sumrate_vs_B.dat};
\addlegendentry{FC-MiLAC-AR}
\addplot[orange!85!black, mark=diamond*, densely dashed, mark options={solid}] table[x expr=\thisrow{B}*1000, y=SCPHP] {data/sumrate_vs_B.dat};
\addlegendentry{SC-MiLAC-PHP}
\addplot[brown, mark=o, mark options={solid, fill=white}, densely dashed] table[x expr=\thisrow{B}*1000, y=FCPHP] {data/sumrate_vs_B.dat};
\addlegendentry{FC-MiLAC-PHP}
\addplot[teal!70!black, densely dashdotted, no marks] table[x expr=\thisrow{B}*1000, y=PSH] {data/sumrate_vs_B.dat};
\addlegendentry{FC-PS-Hybrid, $3$-bit}
\addplot[black!60, densely dotted, no marks] table[x expr=\thisrow{B}*1000, y=PS] {data/sumrate_vs_B.dat};
\addlegendentry{FC-PS, $3$-bit}
\end{axis}
\end{tikzpicture}
\vspace{-3mm}
\caption{Sum-rate versus dynamic range $B$ when $N=64$, $K=4$, ${\mathrm{SNR}=10}$~dB, and $q=3$.}
\label{fig:sr_B}
\end{figure}

Figure~\ref{fig:sr_B} sweeps $B$ at $10$~dB and $q=3$. The \ac{AR} curves first rise and then fall, because a smaller $B$ clips more entries of $\bv^\star$, whereas a larger $B$ makes the grid step $\Delta_{B,L}=2B/(L-1)$ coarser. Both peak at $B=3.5$~mS, which the circuit~\eqref{eq:tac} supports with a symmetric range (Remark~\ref{rem:circuit}), and we use this value as the default. The $3$-bit \ac{PS} baselines do not depend on $B$. \ac{SC}-\ac{MiLAC}-\ac{AR} exceeds the $3$-bit FC-PS for $B$ from $1.75$ to $28$~mS and \ac{FC}-\ac{MiLAC}-\ac{AR} for $B$ from $0.875$ to $56$~mS, whereas the $3$-bit FC-PS-Hybrid stays above both for every $B$.

\begin{figure}[t]
\centering
\begin{tikzpicture}
\begin{axis}[
    width=1\linewidth, height=4.7cm,
    xlabel={Resolution $q$ (bits)}, ylabel={Sum-rate (bits/s/Hz)},
    xmin=1.0, xmax=6.0, ymin=0, ymax=32,
    xtick={1,2,3,4,5,6},
    grid=both, grid style={black!12},
    legend style={font=\scriptsize, at={(0.5,1.03)}, anchor=south, draw=black!30},
    legend cell align={left}, legend columns=3,
    tick label style={font=\scriptsize}, label style={font=\small},
    every axis plot/.append style={line width=0.8pt, mark size=1.6pt}
]
\addplot[purple, no marks, dashed] table[x=q, y=SCUnc] {data/sumrate_vs_q.dat};
\addlegendentry{Unc}
\addplot[red, mark=triangle*] table[x=q, y=SCAR] {data/sumrate_vs_q.dat};
\addlegendentry{SC-MiLAC-AR}
\addplot[teal, mark=square*] table[x=q, y=FCAR] {data/sumrate_vs_q.dat};
\addlegendentry{FC-MiLAC-AR}
\addplot[orange!85!black, mark=diamond*, densely dashed, mark options={solid}] table[x=q, y=SCPHP] {data/sumrate_vs_q.dat};
\addlegendentry{SC-MiLAC-PHP}
\addplot[brown, mark=o, mark options={solid, fill=white}, densely dashed] table[x=q, y=FCPHP] {data/sumrate_vs_q.dat};
\addlegendentry{FC-MiLAC-PHP}
\addplot[teal!70!black, mark=star, dotted, mark options={solid}] table[x=q, y=PSH] {data/sumrate_vs_q.dat};
\addlegendentry{FC-PS-Hybrid}
\addplot[gray, mark=triangle, dotted, mark options={solid}] table[x=q, y=PS] {data/sumrate_vs_q.dat};
\addlegendentry{FC-PS}
\end{axis}
\end{tikzpicture}
\vspace{-3mm}
\caption{Sum-rate versus resolution $q$ when $N=64$, $K=4$, $\mathrm{SNR}=10$~dB, and $B=3.5$~mS. For the \ac{PS} baselines, $q$ is the \ac{PS} resolution.}
\label{fig:sr_q}
\end{figure}

Figure~\ref{fig:sr_q} sweeps $q$ at $10$~dB and $B=3.5$~mS. From $q=1$ to $6$, \ac{SC}-\ac{MiLAC}-\ac{AR} rises from $72.9\%$ to $95.8\%$ of Unc and \ac{FC}-\ac{MiLAC}-\ac{AR} from $85.4\%$ to $99.4\%$, but the gain of \ac{SC}-\ac{MiLAC}-\ac{AR} per added bit falls from $2.9$ to $0.2$~bits/s/Hz, which motivates the default $q=3$. \ac{PHP} gains less than $0.1$~bits/s/Hz beyond $q=3$, since its loss stems from clipping. Both \ac{AR} designs exceed FC-PS for every $q$. FC-PS-Hybrid exceeds both for $q\ge 2$, and \ac{FC}-\ac{MiLAC}-\ac{AR} exceeds it for $q=1$.

\begin{table}[t]
\centering
\caption{Mean runtime in seconds, with the mean number of outer iterations or sweeps in parentheses, and time per \ac{AR} sweep in ms with rank-one updates versus one inversion per candidate (estimated), at $\mathrm{SNR}=10$~dB, $q=3$, and $B=3.5$~mS over $20$ realizations.}
\label{tab:runtime}
\resizebox{\columnwidth}{!}{%
\begin{tabular}{@{}lcccccc@{}}
\toprule
 & Alg.~\ref{alg:stem} & \cite[Alg.~2]{Wu2026MiLAC} & \multicolumn{2}{c}{\ac{AR}} & \multicolumn{2}{c}{Time per sweep} \\
$(N,K)$ & \ac{SC} & \ac{FC} & \ac{SC} & \ac{FC} & \ac{SC} & \ac{FC} \\
\midrule
$(16,4)$ & $0.03$ ($37$) & $0.12$ ($66$) & $0.01$ ($29$) & $0.01$ ($23$) & $0.26$/$18$ & $0.32$/$29$ \\
$(32,4)$ & $0.04$ ($32$) & $0.03$ ($80$) & $0.01$ ($23$) & $0.02$ ($28$) & $0.37$/$80$ & $0.65$/$204$ \\
$(64,4)$ & $0.10$ ($38$) & $0.03$ ($92$) & $0.02$ ($33$) & $0.10$ ($43$) & $0.72$/$376$ & $2.2$/$1710$ \\
$(128,4)$ & $0.33$ ($66$) & $0.06$ ($107$) & $0.08$ ($36$) & $0.74$ ($69$) & $2.1$/$2382$ & $10.8$/$2.0{\cdot}10^4$ \\
$(256,4)$ & $0.57$ ($62$) & $0.02$ ($84$) & $0.36$ ($42$) & $5.73$ ($85$) & $8.5$/$2.1{\cdot}10^4$ & $68$/$3.5{\cdot}10^5$ \\
$(64,2)$ & $0.01$ ($28$) & $0.01$ ($34$) & $0.02$ ($35$) & $0.07$ ($38$) & $0.44$/$189$ & $1.9$/$1619$ \\
$(64,8)$ & $0.42$ ($36$) & $0.05$ ($87$) & $0.07$ ($40$) & $0.17$ ($43$) & $1.7$/$832$ & $3.9$/$2120$ \\
\bottomrule
\end{tabular}}
\end{table}

\begin{figure}[t]
\centering
\begin{tikzpicture}
\begin{axis}[
    width=1\linewidth, height=4.7cm,
    xmode=log, log basis x={10},
    xlabel={Number of outer iterations $n$}, ylabel={Sum-rate (\%)},
    xmin=1, xmax=3000, ymin=96, ymax=100.3,
    grid=both, grid style={black!12},
    legend style={font=\scriptsize, at={(0.98,0.04)}, anchor=south east, draw=black!30}, legend cell align={left}, legend columns=2,
    tick label style={font=\scriptsize}, label style={font=\small},
    every axis plot/.append style={line width=0.8pt},
]
\addplot[blue] table[x=n, y=snr0] {data/sumrate_vs_iter_alg1.dat};
\addlegendentry{$\mathrm{SNR}=0$~dB}
\addplot[teal, densely dashed] table[x=n, y=snr10] {data/sumrate_vs_iter_alg1.dat};
\addlegendentry{$10$~dB}
\addplot[orange!85!black, dashdotted] table[x=n, y=snr20] {data/sumrate_vs_iter_alg1.dat};
\addlegendentry{$20$~dB}
\addplot[red, densely dotted] table[x=n, y=snr25] {data/sumrate_vs_iter_alg1.dat};
\addlegendentry{$25$~dB}
\end{axis}
\end{tikzpicture}
\vspace{-3mm}
\caption{Sum-rate of Algorithm~\ref{alg:stem} after $n$ outer iterations, normalized by its final value and averaged over the channels, when $N=64$ and $K=4$.}
\label{fig:conv}
\end{figure}

\subsection{Complexity and Convergence}
\label{sec:sim_cost}

\textbf{Runtime:} In Table~\ref{tab:runtime}, Algorithm~\ref{alg:stem} takes less than $0.6$~s on average, and~\cite[Alg.~2]{Wu2026MiLAC} is at least as fast for $N\ge 32$, since its iterations operate on the $K\times K$ matrix $\Ym$ instead of the $N\times K$ matrix $\Fm$. Its solutions, however, reach $\St(N,K)$ through the construction of Proposition~\ref{prop:phase} only for $N\ge 2K-1$, whereas Algorithm~\ref{alg:stem} applies for every $N\ge K$. For the \ac{SC}-\ac{MiLAC}, \ac{AR} needs at most $0.36$~s and $42$ sweeps on average. Relative to a sweep with one inversion per candidate, whose time we estimate as $N_CL$ times that of one inversion and one sum-rate evaluation, the rank-one updates shorten each sweep by a factor that grows with $N$, up to about $2500$ (\ac{SC}) and $5100$ (\ac{FC}) at $N=256$. There, \ac{AR} for the \ac{FC}-\ac{MiLAC}, with $16.5$ times as many components, takes about $16$ times as long as for the \ac{SC}-\ac{MiLAC}. From $q=1$ to $6$ at $N=64$, the runtime of \ac{AR} grows from $0.01$ to $0.29$~s (\ac{SC}) and from $0.04$ to $1.13$~s (\ac{FC}), since each sweep scans $2^q$ values per component and the mean number of sweeps grows from about $20$ to $170$. At $q=6$, $17\%$ (\ac{SC}) and $26\%$ (\ac{FC}) of the runs stop at $T_{\mathrm{out}}$.

\textbf{Convergence of Algorithm~\ref{alg:stem}:} In Fig.~\ref{fig:conv}, Algorithm~\ref{alg:stem} attains on average at least $96\%$ of its final sum-rate after one outer iteration. Up to $10$~dB, it settles within about $50$ iterations, whereas at high \ac{SNR} it needs hundreds of iterations, e.g., $706$ on average at $25$~dB, and it never reaches its limit of $5000$.

\subsection{Sensitivity to Channel and Hardware Models}
\label{sec:sim_robust}

\textbf{Channel models:} Table~\ref{tab:robust} compares the default model with Rayleigh fading and with two variants of the default model, with users in a $20^\circ$ sector whose center is uniform in $[-60^\circ,60^\circ]$ or with \ac{NLOS} correlation $0.8^{|m-n|}$ between antennas $m$ and $n$. In every case, \ac{SC}-\ac{MiLAC} and \ac{FC}-\ac{MiLAC} reach the same mean sum-rate within $0.01$~bits/s/Hz, \ac{FC}-\ac{MiLAC}-\ac{AR} exceeds \ac{SC}-\ac{MiLAC}-\ac{AR}, both exceed the $3$-bit FC-PS, and the $3$-bit FC-PS-Hybrid exceeds both. Users in a narrow sector have strongly correlated channels, which widens the gap to Digital from $1.1$ to $2.0$~bits/s/Hz, consistent with~\cite{Wu2026MiLAC}.

\textbf{Grid variations:} For the default model, the varactor grid of $L=8$ equally spaced capacitances between those that give $-B$ and $B$ in~\eqref{eq:tac} attains nearly the same sum-rates as the uniform grid, and the asymmetric range $[-B/2,3B/2]$ of the same width reduces the sum-rates of \ac{SC}-\ac{MiLAC}-\ac{AR} and \ac{FC}-\ac{MiLAC}-\ac{AR} by $1.0$ and $0.4$~bits/s/Hz without changing the ranking of the schemes.

\begin{table}[t]
\centering
\caption{Mean sum-rate (bits/s/Hz) over $200$ realizations for several channel models, with $N=64$, $K=4$, $\mathrm{SNR}=10$~dB, $q=3$, and $B=3.5$~mS. The \ac{PS} baselines use $3$-bit \acp{PS}.}
\label{tab:robust}
\resizebox{\columnwidth}{!}{%
\begin{tabular}{@{}lccccccc@{}}
\toprule
 & Digital & \ac{SC} & \ac{FC} & \ac{SC}-\ac{AR} & \ac{FC}-\ac{AR} & PS-Hyb. & PS \\
\midrule
Rician, $6$~dB (default) & $28.83$ & $27.75$ & $27.75$ & $25.04$ & $26.28$ & $27.96$ & $19.88$ \\
Rayleigh & $28.99$ & $27.71$ & $27.71$ & $25.03$ & $26.35$ & $27.37$ & $19.29$ \\
Users in a $20^\circ$ sector & $27.62$ & $25.64$ & $25.64$ & $23.19$ & $24.28$ & $26.56$ & $18.55$ \\
Correlated \ac{NLOS} & $28.80$ & $27.79$ & $27.79$ & $24.93$ & $26.19$ & $27.95$ & $19.87$ \\
\bottomrule
\end{tabular}}
\end{table}
\section{Conclusion}
\label{sec:conc}

We studied \ac{SC}-\ac{MiLAC}-aided \ac{MU-MISO} downlink beamforming without a digital beamformer and developed a convergent \ac{WMMSE} algorithm on the Stiefel manifold for every $N\ge K$ and, for bounded, discrete susceptances, a refinement that never decreases the sum-rate. Two conclusions follow. First, restricting the design to beamformers with orthonormal columns, which the \ac{SC}-\ac{MiLAC} realizes in closed form~\cite{Nerini2026Stem}, provably loses no sum-rate relative to the \ac{FC}-\ac{MiLAC} for $N\ge 2K-1$, and the \ac{SC} topology reduces the number of tunable components from quadratic to linear in $N$, e.g., by a factor of $4.5$ for $N=64$ and $K=4$. For $N\le 2K-2$, a numerical synthesis indicates that the remaining gap is caused by the Stiefel restriction rather than by the \ac{SC} topology. Second, with $3$-bit susceptances at an \ac{SNR} of $10$~dB, the refinement retains $90\%$ (\ac{SC}) and $95\%$ (\ac{FC}) of the ideal sum-rate, so the \ac{SC}-\ac{MiLAC} trades a moderate sum-rate loss relative to the \ac{FC}-\ac{MiLAC} for far fewer tunable components. At this \ac{SNR}, at least half of the loss of the refinement is due to the starting point of its local search rather than to the grid.

Several directions extend this work. A proof that the \ac{SC}-\ac{MiLAC} realizes the whole spectral-norm ball would close the case $N\le 2K-2$. Better starting points and search methods for the discrete problem, designs that account for component losses and mutual coupling~\cite{ZhouLossy,NeriniPhysics}, channel estimation under the hardware grid~\cite{ZhangEstimation,ZhangMUEstimation}, wideband systems~\cite{Peng2026OFDM}, and experimental validation~\cite{NeriniHardware} are natural next steps.

\appendix[Proof of Proposition~\ref{prop:alg1_conv}(b)]
\label{app:conv}

We write $\xv\triangleq\sqrt{\pv}$ for the amplitudes, $\mathcal{X}\triangleq\{\xv\ge\zerov:\|\xv\|^2\le P_T\}$ for their feasible set, and $g(\uv,\omegav,\xv,\Fm)$ for the objective of~\eqref{eq:wmmse_stem} with $\sqrt{p_k}$ replaced by $x_k$. At outer iteration $n$, Algorithm~\ref{alg:stem} sets $(\uv^{(n)},\omegav^{(n)})=U(\xv^{(n-1)},\Fm^{(n-1)})$, where $U$ denotes the closed forms~\eqref{eq:u_update_stem}, then $\xv^{(n)}\in\argmin_{\xv\in\mathcal{X}}g_n(\xv)$ with $g_n(\xv)\triangleq g(\uv^{(n)},\omegav^{(n)},\xv,\Fm^{(n-1)})$ by~\eqref{eq:p_update_stem}, and then $\Fm^{(n)}$ by the inner loop started at $\Fm^{(n-1)}$.

\emph{Step~1 (compactness):} $\Fm^{(n)}\in\St(N,K)$ and $\xv^{(n)}\in\mathcal{X}$ lie in compact sets. By~\eqref{eq:u_update_stem}, $|u_k|\le\sqrt{P_T}\|\hv_k\|/\sigma^2$, and at the \ac{MMSE} receiver $e_k\in[\sigma^2/(P_T\|\hv_k\|^2+\sigma^2),1]$, so $1\le\omega_k\le 1+P_T\|\hv_k\|^2/\sigma^2$. Hence all iterates lie in a compact set $\mathcal{C}$, on which $g$ is smooth and bounded below.

\emph{Step~2 (sufficient decrease):} The $\xv$-update decreases $g$ by $d_n\triangleq g_n(\xv^{(n-1)})-g_n(\xv^{(n)})\ge 0$. Let $\mathbf{G}^{(n)}$ denote the Riemannian gradient $\mathrm{grad}_{\Fm}\,g(\uv^{(n)},\omegav^{(n)},\xv^{(n)},\Fm^{(n-1)})$ at the start of the inner loop. If $\mathbf{G}^{(n)}=\Zerom$, the inner loop keeps $\Fm^{(n-1)}$, and~\eqref{eq:app_F} below holds trivially. Otherwise, its first iteration moves along $\Xim=-\mathbf{G}^{(n)}/\|\mathbf{G}^{(n)}\|_F$ with an Armijo step, and the later iterations do not increase $g$. The function $h(\alpha)\triangleq g(\uv^{(n)},\omegav^{(n)},\xv^{(n)},R_{\Fm^{(n-1)}}(\alpha\Xim))$ is smooth, since ${\Fm^{(n-1)}+\alpha\Xim}$ has full column rank for every $\alpha$, as its Gram matrix is $\Id_K+\alpha^2\Xim^\HH\Xim$, and $h'(0)=-\|\mathbf{G}^{(n)}\|_F$. As $(\uv,\omegav,\xv)$, $\Fm$, the unit tangent vector $\Xim$, and $\alpha\in[0,\bar\alpha]$ range over compact sets, $|h''(\alpha)|\le L_h$ for a constant $L_h$. Hence the Armijo condition holds for all $\alpha\le 2(1-c_{\mathrm{A}})\|\mathbf{G}^{(n)}\|_F/L_h$, the backtracking accepts $\alpha\ge\min\{\bar\alpha,2\rho(1-c_{\mathrm{A}})\|\mathbf{G}^{(n)}\|_F/L_h\}$, and the $\Fm$-update decreases $g$ by at least
\begin{equation}\label{eq:app_F}
\psi(\|\mathbf{G}^{(n)}\|_F)\triangleq c_{\mathrm{A}}\|\mathbf{G}^{(n)}\|_F\min\Bigl\{\bar\alpha,\frac{2\rho(1-c_{\mathrm{A}})\|\mathbf{G}^{(n)}\|_F}{L_h}\Bigr\}.
\end{equation}

\emph{Step~3 (limits):} Since the updates of $(\uv,\omegav)$ do not increase $g$ and $g$ is bounded below on $\mathcal{C}$, the sum of $d_n+\psi(\|\mathbf{G}^{(n)}\|_F)$ over all $n$ is finite, so $d_n\to 0$ and $\mathbf{G}^{(n)}\to\Zerom$. Let $(\overline{\xv},\overline{\Fm})$ be the limit of a subsequence $(\xv^{(n_j-1)},\Fm^{(n_j-1)})$, and, after passing to a further subsequence, let $\xv^{(n_j)}\to\widehat{\xv}$. By continuity of $U$, $(\uv^{(n_j)},\omegav^{(n_j)})\to(\overline{\uv},\overline{\omegav})\triangleq U(\overline{\xv},\overline{\Fm})$. Let $\varphi(\xv)\triangleq g(\overline{\uv},\overline{\omegav},\xv,\overline{\Fm})$. Since $g_{n_j}(\xv^{(n_j)})\le g_{n_j}(\xv)$ for every $\xv\in\mathcal{X}$, continuity of $g$ shows that $\widehat{\xv}$ minimizes $\varphi$ over $\mathcal{X}$, and $d_{n_j}\to 0$ gives $\varphi(\overline{\xv})=\varphi(\widehat{\xv})$, so $\overline{\xv}$ also minimizes $\varphi$ over $\mathcal{X}$. Moreover, $\mathrm{grad}_{\Fm}\,g(\overline{\uv},\overline{\omegav},\widehat{\xv},\overline{\Fm})=\lim_j\mathbf{G}^{(n_j)}=\Zerom$.

\emph{Step~4 (gradient at $\overline{\xv}$):} Up to an additive constant, $\varphi(\xv)=\sum_k(\overline{\eta}_kx_k^2-2\overline{\nu}_kx_k)$, where $\overline{\eta}_k$ and $\overline{\nu}_k$ are given by~\eqref{eq:a_stem} at $(\overline{\uv},\overline{\omegav},\overline{\Fm})$. If $\overline{\eta}_k>0$, then $\overline{x}_k=\widehat{x}_k$, since otherwise the midpoint of the two minimizers would give a smaller value of $\varphi$. If $\overline{\eta}_k=0$, then $\overline{u}_i\hv_i^\HH\overline{\fv}_k=0$ for every $i$, which yields $\overline{u}_k=0$ because $\overline{u}_k$ is proportional to $\hv_k^\HH\overline{\fv}_k$ by~\eqref{eq:u_update_stem}. By~\eqref{eq:Egrad_F}, the $k$th column of $\nabla_{\Fm}g(\overline{\uv},\overline{\omegav},\xv,\overline{\Fm})$ is $-2x_k\bigl(\overline{\omega}_k\overline{u}_k\hv_k-x_k\sum_i\overline{\omega}_i|\overline{u}_i|^2(\hv_i^\HH\overline{\fv}_k)\hv_i\bigr)$, which depends on $\xv$ only through $x_k$ and vanishes for every $x_k$ if $\overline{\eta}_k=0$. Hence the Euclidean, and thus the Riemannian, gradients with respect to $\Fm$ at $\overline{\xv}$ and at $\widehat{\xv}$ coincide, and $\mathrm{grad}_{\Fm}\,g(\overline{\uv},\overline{\omegav},\overline{\xv},\overline{\Fm})=\Zerom$.

\emph{Step~5 (stationarity for~\eqref{eq:Pstem}):} By~\eqref{eq:wmmse_identity_stem}, $g(\overline{\uv},\overline{\omegav},\xv,\Fm)\ge\min_{\uv,\omegav}g(\uv,\omegav,\xv,\Fm)=K-\mathcal{R}(\Fm\diag(\xv))\ln 2$ for all ${(\xv,\Fm)\in\mathbb{R}^K\times\mathbb{C}^{N\times K}}$, with equality at $(\overline{\xv},\overline{\Fm})$, because $(\overline{\uv},\overline{\omegav})=U(\overline{\xv},\overline{\Fm})$. The smooth function $D(\xv,\Fm)\triangleq g(\overline{\uv},\overline{\omegav},\xv,\Fm)-K+\mathcal{R}(\Fm\diag(\xv))\ln 2$ is therefore nonnegative and vanishes at $(\overline{\xv},\overline{\Fm})$, which is thus an unconstrained minimizer of $D$, so its Euclidean gradients vanish there, i.e., the gradients of $g(\overline{\uv},\overline{\omegav},\cdot,\cdot)$ equal $-\ln 2$ times those of $\mathcal{R}$ at $(\overline{\xv},\overline{\Fm})$. With Step~4, this gives $\mathrm{grad}_{\Fm}\mathcal{R}(\overline{\Fm}\diag(\overline{\xv}))=\Zerom$. Moreover, $\overline{\xv}$ minimizes the convex function $\varphi$ over $\mathcal{X}$, so it satisfies the \ac{KKT} conditions with the gradient of $\varphi$, which equals $-\ln 2$ times the gradient of $\mathcal{R}(\overline{\Fm}\diag(\cdot))$. These are the \ac{KKT} conditions of $\max_{\xv\in\mathcal{X}}\mathcal{R}(\overline{\Fm}\diag(\xv))$ at $\overline{\xv}$.\hfill$\square$

\bibliography{references}
\bibliographystyle{IEEEtran}

\end{document}